\documentclass[a4paper,10pt]{amsart}

\usepackage{amssymb,epsfig,amsmath,amsthm}
\usepackage[font=small,format=plain,labelfont=bf,up,textfont=it,up]{caption}
\usepackage{graphicx,subfig}
\graphicspath{{./Figs/}}
\numberwithin{equation}{section}
\DeclareMathOperator*{\argmin}{arg\,min}

\DeclareMathOperator*{\argmax}{arg\,max}

\title[Solving HJB Equations and HJB Obstacle Problems]{Penalty Methods for the Solution of Discrete HJB Equations --
Continuous Control and Obstacle Problems}

\author{J. H. Witte}
\author{C. Reisinger}
\address{Mathematical Institute\\
University of Oxford}
\date{November 2011}
\email{\texttt{[\,witte\,,\,reisinge\,]\,@\,maths.ox.ac.uk}}
\thanks{JHW acknowledges support from Balliol College, University of Oxford, the
UK Engineering and Physical Sciences Research Council (EPSRC), and the Oxford-Man Institute of Quantitative Finance, University of Oxford.
We thank Guy Barles, Laboratoire de Math\'{e}matiques et Physique Th\'{e}orique, Tours, for many helpful comments.}

\begin{document}

\begin{abstract}
In this paper, we present a novel penalty approach for the numerical solution of continuously controlled HJB equations and HJB obstacle problems. Our results include estimates of the penalisation error for a class of penalty terms, and we show that
variations of Newton's method can be used to obtain globally convergent iterative  solvers for the penalised equations. Furthermore, we discuss under what conditions local quadratic convergence of the iterative solvers can be expected. We include numerical results demonstrating the competitiveness of
our methods.\bigskip

%
%

\hspace{-.45cm}\textit{Key Words:} HJB Equation, HJB Obstacle Problem, Min-Max Problem, Numerical Solution, Penalty Method,
Semi-Smooth Newton Method, Viscosity Solution\bigskip

\hspace{-.45cm}\textit{2010 Mathematics Subject Classification:} 65M12, 93E20\bigskip

\end{abstract}

\maketitle

\newtheorem{theorem}{Theorem}[section]
\newtheorem{la}[theorem]{Lemma}
\newtheorem{cor}[theorem]{Corollary}
\newtheorem{remark}[theorem]{Remark}
\newtheorem{prob}[theorem]{Problem}	
\newtheorem{definition}[theorem]{Definition}
\newtheorem{alg}[theorem]{Algorithm}
\newtheorem{prop}[theorem]{Proposition}
\newtheorem{assumption}[theorem]{Assumption}

\section{Introduction}

Problems of optimal stochastic control arise numerously in the mathematical analysis of real-world phenomena (cf.\,\cite{Karatzas_MethodsMathsFinance,Shreve_ContTimeModels}), and,
applying Bellman's principle of optimality, they can often be reformulated as Hamilton-Jacobi-Bellman (HJB) equations (cf.\,\cite{XYZ_StochControlHJBBook}).
Numerical approaches for solving HJB equations can roughly be split into two fields: Markov chain
approximations (e.g.\,cf.\,\cite{Kushner_NumMethodsStochControl_ContTime, KushnerDupuis_StochasticControlProblems_ContTime, Fleming_Soner_ControlledMarkovProcesses}) and finite difference
methods (see \cite{BarlesMainArticle, Barles_ErrorBoundsMonotoneApproxSchemes, Barles_OntheConvergenceRate_ApproxHJBEq} and references therein).
Even though both approaches can be found to be closely related (e.g.\,cf.\,\cite{Song_MarkovChainHJBEq}), there is a rather clear
conceptual difference: Markov chain approximations make explicit use of the underlying stochastic
model and solve for transition densities, while finite difference methods deal solely with
the HJB equation and solve for the unique viscosity solution directly (cf.\,\cite{UsersGuide_ViscositySols}).\bigskip

We begin by introducing the two kinds of non-linear equations we aim to solve numerically in this paper.
Initially, the problem formulations are slightly informal and primarily motivational, and any equation of the described type permitting a discretisation with certain properties fits into our framework.

\begin{prob}\label{DifferentialOperatorProblem}
Let $\mathbf{U}\subset\mathbb{R}$ be a compact\footnote{Our results can trivially be extended to any finite union of compact intervals in $\mathbb{R}^N$, which, in particular, includes any finite set. For simplicity of presentation, we work with $\mathbf{U}$ as given.} interval. Let $\Omega\subset\mathbb{R}^n$ (with $n\in\mathbb{N}$) be an open set, and let $\mathcal{L}_u$\,, $u\in\mathbf{U}$, be a given family of affine differential operators on $\Omega$. Find a function $V:\Omega\to\mathbb{R}$ such that
\begin{equation}
\inf_{u\in\mathbf{U}}\{\mathcal{L}_uV\}=0.\label{DifferentialOperatorProblem_Eq1}
\end{equation}
\end{prob}

\begin{prob}\label{DifferentialOperatorProblem_MinMax}
In the situation of Problem \ref{DifferentialOperatorProblem}, let $\tilde{\mathcal{L}}$ be an additional affine differential operator. Find a function $W:\Omega\to\mathbb{R}$ such that
\begin{equation}
\inf\Big\{\sup_{u\in\mathbf{U}}\{\mathcal{L}_uW\},\,\tilde{\mathcal{L}}W\Big\}=0.\label{DifferentialOperatorProblem_MinMax_Eq1}
\end{equation}
\end{prob}

Equation \eqref{DifferentialOperatorProblem_Eq1} is a standard HJB equation as arising from many problems of stochastic optimal control (cf.\,\cite{XYZ_StochControlHJBBook,Pham_StochControl}); examples in which the control set $\mathbf{U}$ is truly infinite and does not reduce to a finite
set include problems of indifference pricing of financial derivatives (cf.\,\cite{Carmona_IndifferencePricing_Book} and references therein) and gas storage valuation (cf.\,\cite{Thompson_NaturalGasStorageValuation,Chen_Forsyth_NaturalGasStorageValuation}). The second formulation, Problem \ref{DifferentialOperatorProblem_MinMax}, is an obstacle problem (cf.\,\cite{ElliotOckendon}) involving an HJB equation, and is a special case of an Isaacs' equation (cf.\,\cite{XYZ_StochControlHJBBook,UsersGuide_ViscositySols,Isaac_DifferentialGames}); considering an example from \cite{Oberman_Thalia_EarlyExerciseInbcompleteMarkets}, we will later see that the computation of an early exercise indifference price in an incomplete market model leads to an equation of the form \eqref{DifferentialOperatorProblem_MinMax_Eq1}, a context in which the similarity to an obstacle problem can clearly be traced back to the American exercise feature (cf.\,\cite{OptionPricing}).\bigskip

Evidently, if we replace \eqref{DifferentialOperatorProblem_Eq1} by
\begin{equation}
\widehat{\mathcal{L}}V + \inf_{u\in\mathbf{U}}\{\mathcal{L}_uV\}=0,\label{DifferentialOperatorProblem_ExtraOperator_Eq1}
\end{equation}
we can easily recover formulation \eqref{DifferentialOperatorProblem_Eq1} by moving $\widehat{\mathcal{L}}V$ inside the ``inf'' and defining $\bar{\mathcal{L}}_u:=\mathcal{L}_u+\widehat{\mathcal{L}}$, $u\in\mathbf{U}$; in some applications, the operator $\widehat{\mathcal{L}}V$ can be found to simply be the time derivative, i.e. $\widehat{\mathcal{L}}V=\partial V/\partial t$, whereas in other cases it may constitute the main part of a considered equation.\bigskip

In this paper, we are concerned with the finite difference approximation of equations \eqref{DifferentialOperatorProblem_Eq1} and \eqref{DifferentialOperatorProblem_MinMax_Eq1}.
Based on the -- most helpful -- results in \cite{BarlesMainArticle, Barles_ErrorBoundsMonotoneApproxSchemes, Barles_OntheConvergenceRate_ApproxHJBEq},
it can be shown, for most HJB equations, that a monotone, stable, and consistent finite difference approximation converges to the unique viscosity solution, which is the correct notion of solution since it usually corresponds to the stochastic model.
However, even if convergence can be guaranteed, one still has to solve the resulting non-linear discrete systems. The
only exception are fully explicit time-stepping schemes, which generally suffer from undesirable stability constraints.\bigskip

We assume a rather general class of discretisations, one that is likely to arise when applying the convergence results cited above, and
we study how the discrete systems can be solved using penalty methods.\bigskip

Existing work
on the topic includes \cite{Forsyth_Controlled_HJB_PDEs_Finance,Wang_Forsyth_MaximalUSeCentralDifferences_HJBFinance,Forsyth_NonlinearPDEFinance,Bokanowski_Howard_Algorithm,Forsyth_CombinedFixedPointIteration}, where a policy iteration algorithm is used,
and \cite{WitteReisinger_PenaltyScheme_DiscreteControlledHJBEquations}, where a penalty approximation for HJB equations is introduced.
For the solution of the penalised equations, we consider Newton-like iterative solvers, many properties of which we prove based on results in \cite{Pang_NewtonsMethod,Qi_NonsmoothNewton}. See also \cite{HintermullerItoKunisch_PrimalDualAcitveSetStrategy_SemismoothNewtonMethod,ItoKunisch_SemiSmoothNewton_VarIneqFirstKind,ItoKunisch_SemiSmoothNewtonAndGlobalization} regarding the use of semi-smooth Newton methods in the context of variational inequalities and their interpretation as primal-dual active set strategies, which in turn are closely related to the method of policy iteration.\bigskip

We extend the previous
results in two ways. First, we consider a compact set of controls; until now, except for \cite{Bokanowski_Howard_Algorithm}, all algorithms were based on
the assumption of a finite control set. Second, we also solve an HJB obstacle problem, which arises for example in mathematical finance when pricing early
exercise claims in incomplete markets, e.g.\ in \cite{Oberman_Thalia_EarlyExerciseInbcompleteMarkets}; to the best of our knowledge, the only existing algorithm capable of handling this problem is the Ho-3 algorithm introduced in \cite{Bokanowski_Howard_Algorithm}, which can be considered as a nested policy iteration. The discussed techniques have many applications in mathematical finance. Besides the applications already considered in \cite{Forsyth_Controlled_HJB_PDEs_Finance,Forsyth_NonlinearPDEFinance,WitteReisinger_PenaltyScheme_DiscreteControlledHJBEquations}, which include uncertain volatility models, transaction cost models, and unequal borrowing/lending rates and stock borrowing fees, our methods are applicable to many problems from portfolio optimisation, including indifference pricing \cite{Zariphopoulou_UnhedgeableRisk,Carmona_IndifferencePricing_Book}
and indifference pricing with early exercise features \cite{Oberman_Thalia_EarlyExerciseInbcompleteMarkets}.

\subsection*{Structure of this Paper} We aim to devise numerical schemes for the solution of HJB equations
and HJB obstacle problems. For a general overview of our numerical approach, we refer to
Figure 1 in \cite{WitteReisinger_PenaltyScheme_DiscreteControlledHJBEquations}, where a similar conceptual structure is depicted.\medskip

\paragraph{\it{Section \ref{ProblemFormulation}}} 
We relate the two non-linear equations of Problems \ref{DifferentialOperatorProblem} and \ref{DifferentialOperatorProblem_MinMax}
each to a non-linear discrete system; the latter are chosen such that they are monotone and correspond to what naturally arises when applying a fully implicit or weighted
time stepping discretisation to the equations. Since implicit schemes are usually unconditionally stable and
finite difference schemes are naturally consistent, this setup allows to conclude convergence to the unique viscosity solution whenever a strong
comparison principle holds (cf.\,\cite{BarlesMainArticle, Barles_ErrorBoundsMonotoneApproxSchemes, Barles_OntheConvergenceRate_ApproxHJBEq}).
We proof existence and uniqueness of solutions to the discrete problems and present some further properties.
\medskip

\paragraph{\it{Section \ref{Section_PenalisationDiscrProb}}} We present a penalised problem which approximates the discrete HJB equation
to an accuracy of $O(1/\rho)$, where $\rho>0$ is the penalty parameter. A similar approach was used in \cite{ForsythQuadraticConvergence}
for American options, and was extended to HJB equations in \cite{WitteReisinger_PenaltyScheme_DiscreteControlledHJBEquations}. The novelty
of the approach described in this paper is that we do not penalise for every control individually, as was done in \cite{WitteReisinger_PenaltyScheme_DiscreteControlledHJBEquations},
but only penalise the maximum violation; this significantly simplifies the penalty formulation and allows for compact control sets.\medskip

\paragraph{\it{Section \ref{SectionModNewton}}}
We study the iterative solution of the penalised HJB equation presented in Section \ref{Section_PenalisationDiscrProb}.
We introduce two globally convergent Newton-type solvers, and we show that, when using a smooth penalty term, the classical Newton scheme can be expected to have local quadratic convergence.\medskip

\paragraph{\it{Sections \ref{Penalise_Discr_Obst_prob} and \ref{SectionModNewton_Isaacs}}} In these two sections, we extend the techniques from Sections \ref{Section_PenalisationDiscrProb} and \ref{SectionModNewton} to show that the use of penalisation (Section \ref{Penalise_Discr_Obst_prob}) with subsequent non-linear iteration (Section \ref{SectionModNewton_Isaacs}) is also a powerful strategy for the solution of the obstacle version of a continuously controlled HJB equation. In particular, in Section \ref{SectionModNewton_Isaacs}, we present two Newton-type iterative methods, one that converges globally and one that converges locally quadratically.\medskip

\paragraph{\it{Section \ref{ContinuousControl_InPractice}}}
The numerical strategies developed in Sections \ref{Section_PenalisationDiscrProb}--\ref{SectionModNewton_Isaacs} are designed for
HJB equations and HJB obstacle problems with a compact set of controls and continuous dependence of the differential operators on the controls.
For some equations, it may be convenient to approximate the dependence on the controls, e.g.\ by piecewise linearisation; in this section, we prove that the previously introduced algorithms are stable with respect to such perturbations.\medskip

\paragraph{\it{Section \ref{Part_Numerics}}} We conclude by presenting numerical results, solving an HJB equation
taken from \cite{Zariphopoulou_UnhedgeableRisk} and an HJB obstacle problem found in \cite{Oberman_Thalia_EarlyExerciseInbcompleteMarkets}; in both cases, we compare our approach to the method of policy iteration (cf.\,\cite{Forsyth_Controlled_HJB_PDEs_Finance,Bokanowski_Howard_Algorithm}).

\section{Problem Formulation}\label{ProblemFormulation}

Following \cite{FiedlerSpecialMatrices}, we introduce $Z_N$\,, $N\in\mathbb{N}$, to be the set of all real square matrices of dimension $N$ whose off-diagonal entries are all non-positive, i.e.
\begin{equation*}
Z_N:=\{A=(a_{rs})\in\mathbb{R}^{N\times N} : a_{rs}\leq 0,\,r\neq s\},
\end{equation*}
and define
\begin{equation*}
K_N:=\{A\in Z_N : \exists\ x\in\mathbb{R}^N,\, x\geq 0,\,\text{ s.t. } Ax>0\}
\end{equation*}
to be the set of M-matrices in $\mathbb{R}^N$; it can be found in \cite{FiedlerSpecialMatrices} that $A^{-1}\geq 0$ for all matrices $A\in K_N$. Furthermore, we introduce the set
\begin{equation*}
K^o_N:=\Big\{A=(a_{r,s})\in K_N : a_{ss} > \sum_{r\neq s}|a_{rs}|\Big\}\subset K_N\,,
\end{equation*}
which has the following properties.
\begin{itemize}
\item If, for $1\leq i\leq N$, we replace the $i$-th row of $A\in K^o_N$
by the $i$-th row of some $B\in K^o_N$\,, the resulting matrix will still be in $K^o_N$\,.
\item $A+B\in K^o_N$ for two matrices $A$, $B\in K^o_N$\,.
\end{itemize}
Throughout this paper, the set $K^o_N$ will be one of our main building blocks since it allows us to deduce the M-matrix property of a matrix that was constructed from a number of other matrices.

\begin{remark}\label{RelationalOperators_Remark}
Let $N\in\mathbb{N}$ and define $\mathcal{N}:=\{1,\ldots,N\}$.
Consider two vectors $x$, $y\in\mathbb{R}^N$ and a matrix $A\in\mathbb{R}^{N\times N}$. For any $i\in\mathcal{N}$, we denote by $(x)_i$ and $(A)_i$ the $i$-th coordinate and the $i$-th row of vector $x$ and matrix $A$, respectively. When writing $x\geq 0$, we generally mean that $(x)_i\geq 0$ for all $i\in\mathcal{N}$, and we take $z:=\min\{x,y\}$ to be the vector satisfying $(z)_i=\min\{(x)_i\,,(y)_i\}$ for all $i\in\mathcal{N}$.
The definitions extend trivially to other relational operators and to the maximum of two vectors.
\end{remark}

Making use of the definition of $K^o_N$\,, we assume, for now, that we can find sensible discretisations of Problems \ref{DifferentialOperatorProblem} and \ref{DifferentialOperatorProblem_MinMax} of the following form; we will give a more rigorous justification of this assumption later on.

\begin{prob}\label{DiscreteProbDef}
Let
\begin{align}
\mathfrak{A}&:\mathbf{U}\to \mathbb{R}^{N\times N} :u\mapsto A_u\label{DiscreteProbDef_ConFunDef2}\\
\text{and}\quad\mathfrak{B}&:\mathbf{U}\to\mathbb{R}^N :u\mapsto b_u\label{DiscreteProbDef_ConFunDef1}
\end{align}
be continuous functions, with $A_u\in K^o_N$\,, $u\in\mathbf{U}$. Find $x\in\mathbb{R}^N$ such that
\begin{equation}
\min_{u\in\mathbf{U}}\{A_u\,x-b_u\}=0.\label{DiscreteProbDef_Eq1}
\end{equation}
\end{prob}

The continuity requirement of $\mathfrak{A}$ and $\mathfrak{B}$ above is to be understood in the sense of any vector norm, say $\|\cdot\|_2$ or $\|\cdot\|_\infty$\,, in $\mathbb{R}^{N\times N}$ and $\mathbb{R}^N$. For example, for $N=2$, the maps
$\mathfrak{A}$ and $\mathfrak{B}$ will be of the form
\begin{equation*}
u\mapsto
\left( \begin{array}{cc}
f_1(u) & f_3(u) \\
f_2(u) & f_4(u) \end{array} \right)
\quad\text{and}\quad u\mapsto
\left( \begin{array}{c}
g_1(u) \\
g_2(u) \end{array} \right),
\end{equation*}
respectively, with $f_1$\,, $f_2$\,, $f_3$\,, $f_4$\,, $g_1$\,, $g_2$ $\in C(\mathbf{U},\mathbb{R})$.

\begin{prob}\label{DiscreteProbDef_MinMax}
In the setting of Problem \ref{DiscreteProbDef}, let also $\tilde{A}\in K^o_N$ and $\tilde{b}\in\mathbb{R}^N$. Find $z\in\mathbb{R}^N$ such that
\begin{equation}
\min\Big\{ \max_{u\in\mathbf{U}}\{A_u\,z-b_u\},\,\tilde{A}z-\tilde{b}\Big\}=0.\label{DiscreteProbDef_MinMax_Eq1}
\end{equation}
\end{prob}

We will, without always mentioning so explicitly, make frequent use of the fact that every function which is continuous on a compact interval attains its minimum and maximum on the interval.

\begin{theorem}\label{BothDiscreteProblems_UniqSolvable}
The discrete Problems \ref{DiscreteProbDef} and \ref{DiscreteProbDef_MinMax} are both uniquely solvable.
\end{theorem}
\begin{proof}
Corollary \ref{PenaltyConvergenceToTrueSol} and Theorem \ref{LineSearchNewton_HJB} give the existence of a solution to Problem \ref{DiscreteProbDef}, and, similarly, Corollary \ref{PenaltyConvergenceToTrueSol_MinMax} and Theorem \ref{LineSearchNewton} (or, alternatively, Ho-3 in \cite{Bokanowski_Howard_Algorithm}) give the existence of a solution to Problem \ref{DiscreteProbDef_MinMax}.
Hence, it remains to show the uniqueness of the solutions. To this end, suppose we have two solutions $x_1$ and $x_2$ to Problem \ref{DiscreteProbDef}.
For every $i\in\mathcal{N}$, there exists a $u_i\in\mathbf{U}$ such that
\begin{equation*}
(A_{u_i})_i\,x_2-(b_{u_i})_i=0,
\end{equation*}
and we also have
\begin{equation*}
(A_{u_i})_i\,x_1-(b_{u_i})_i\geq0.
\end{equation*}
Denote by $A^*\in\mathbb{R}^{N\times N}$ the matrix consisting of the rows $(A_{u_i})_i$\,, $i\in\mathcal{N}$. We have $A^*\in K^o_N$ and $A^*(x_1-x_2)\geq 0$, from which we get $x_1-x_2\geq 0$ since $(A^*)^{-1}\geq 0$.
Conversely, using the same arguments but swapping $x_1$ and $x_2$\,, we can also get $x_2-x_1\geq 0$, which then proves the uniqueness of a solution to Problem \ref{DiscreteProbDef}.
Now, suppose we have two solutions $z_1$ and $z_2$ to Problem \ref{DiscreteProbDef_MinMax}.
For $i\in\mathcal{N}$, let $u_i^1$\,, $u_i^2\in\mathbf{U}$ be such that
\begin{equation}
\max_{u\in\mathbf{U}}(A_u\,z_j-b_u)_i=(A_{u^j_i}\,z_j-b_{u^j_i})_i\,,\quad j\in\{1,2\}.\label{DiscreteProb_MinMax_UniqTh_Eq0.5}
\end{equation}
From \eqref{DiscreteProbDef_MinMax_Eq1}, we then get that, for $i\in\mathcal{N}$, we have
\begin{equation}
\min\{ (A_{u^1_i}\,z_1-b_{u^1_i})_i\,,(\tilde{A}z_1-\tilde{b})_i\}-\min\{ (A_{u^2_i}\,z_2-b_{u^2_i})_i\,,(\tilde{A}z_2-\tilde{b})_i\}=0.\label{DiscreteProb_MinMax_UniqTh_Eq1}
\end{equation}
We distinguish the following possible cases in \eqref{DiscreteProb_MinMax_UniqTh_Eq1}.
\begin{enumerate}
\item[(i)] We have $(\tilde{A}z_1-\tilde{b})_i - (\tilde{A}z_2-\tilde{b})_i = 0$.
\item[(ii)] We have $(\tilde{A}z_1-\tilde{b})_i -(A_{u^2_i}\,z_2-b_{u^2_i})_i=0$, in which case $(A_{u^1_i}\,z_1 - A_{u^1_i}\,z_2)_i\geq (\tilde{A}z_1-\tilde{b})_i -(A_{u^2_i}\,z_2-b_{u^2_i})_i=0$.
\item[(iii)] We have $(A_{u^1_i}\,z_1-b_{u^1_i})_i-(\tilde{A}z_2-\tilde{b})_i=0$, in which case $(\tilde{A}z_1 - \tilde{A}z_2)_i\geq (A_{u^1_i}\,z_1-b_{u^1_i})_i - (\tilde{A}z_2-\tilde{b})_i =0$.
\item[(iv)] We have $(A_{u^1_i}\,z_1-b_{u^1_i})_i-(A_{u^2_i}\,z_2-b_{u^2_i})_i=0$, in which case $(A_{u^1_i}\,z_1-A_{u^1_i}\,z_2)_i \geq (A_{u^1_i}\,z_1-b_{u^1_i})_i - (A_{u^2_i}\,z_2-b_{u^2_i})_i =0$.
\end{enumerate}
At this point, we can show the uniqueness of a solution to Problem \ref{DiscreteProbDef_MinMax} by arguing analogously to the first half of this proof (where we introduced the matrix $A^*$ and used its $M$-matrix properties).
\end{proof}

Next, we show that the discretisation matrices appearing in \eqref{DiscreteProbDef_Eq1} and \eqref{DiscreteProbDef_MinMax_Eq1} 
have some convenient boundedness properties.

\begin{la}\label{BoundednessofMatrixandInv_La}
For $\mathfrak{A}$ as in Problem \ref{DiscreteProbDef}, there exists a constant $C>0$ such that
\begin{equation*}
0< \|A_u\|_{\infty}\,,\|A^{-1}_u\|_{\infty}\leq C,\quad u\in\mathbf{U}.
\end{equation*}
\end{la}
\begin{proof}
Since continuity of $\mathfrak{A}$ implies directly the existence of a constant $C_1>0$ such that $\|A_u\|_{\infty}\leq C_1$ for all $u\in\mathbf{U}$, it only remains to prove the corresponding estimate for the inverse matrices.
An M-matrix is, in particular, non-singular with positive determinant (cf.\,\cite{FiedlerSpecialMatrices}), and, hence, we have that $u\mapsto\det(A_u)>0$ is a continuous and positive function on $\mathbf{U}$. As every continuous function on a compact interval assumes its minimum, there exists an $\epsilon>0$ such that
$\det(A_u)>\epsilon$ for all $u\in\mathbf{U}$. Now, in \cite{FiedlerSpecialMatrices}, it can be found that,
for a square non-singular matrix $A=(a_{rs})\in\mathbb{R}^{N\times N}$, it is $A^{-1}=(\alpha_{rs})$ if we define
\begin{equation*}
\alpha_{rs}:=\frac{A_{sr}}{\det(A)}\,,\quad 1\leq r,s\leq N,
\end{equation*}
where $A_{sr}$ is the cofactor of $a_{sr}$ in the matrix $A$. Since the calculation of the cofactor is, like the determinant, a continuous operation on the entries of the matrix, there exists a constant $C_2>0$ such that, for $u\in\mathbf{U}$ and $(A_u)^{-1}=\big((\alpha_u)_{rs}\big)$, it is
\begin{equation*}
|(\alpha_u)_{rs}|=\Big|\frac{(A_u)_{sr}}{\det(A_u)}\Big|\leq\Big|\frac{(A_u)_{sr}}{\epsilon}\Big|\leq \frac{C_2}{\epsilon}\,,\quad 1\leq r,s\leq N.
\end{equation*}
This proves the lemma.
\end{proof}

\begin{cor}\label{BoundednessofMatrixandInv_Cor}
Suppose that, in the situation of Problem \ref{DiscreteProbDef}, we define
\begin{align*}
&\mathcal{U}:= \underbrace{\mathbf{U}\times \mathbf{U}\times\ldots\times \mathbf{U}}_{N\text{-}times}\subset\mathbb{R}^N\\
\text{and}\quad&\mathfrak{A}^{\mathcal{U}}:\mathcal{U}\to\mathbb{R}^{N\times N}:(u_1\,,u_2\,,\ldots,u_N)\mapsto A_{u_1u_2\ldots u_N}\,,
\end{align*}
where $A_{u_1u_2\ldots u_N}$ denotes the matrix having as $i$-th row, $i\in\mathcal{N}$, the $i$-th row of
$A_{u_i}$. In this case, $\mathcal{U}$ is a compact set, $\mathfrak{A}^{\mathcal{U}}$ is a continuous function, and $A_{u_1u_2\ldots u_N}\in K^o_N$ for every $(u_1\,,u_2\,,\ldots,u_N)\in\mathcal{U}$.
Furthermore, there exists a constant $C>0$ such that
\begin{equation*}
\|A_{u_1u_2\ldots u_N}\|_{\infty} +\|A^{-1}_{u_1u_2\ldots u_N}\|_{\infty}
\leq C,\quad (u_1\,,u_2\,,\ldots,u_N)\in\mathcal{U}.
\end{equation*}
\end{cor}
\begin{proof}
The proof is identical to the one of Lemma \ref{BoundednessofMatrixandInv_La}, the only difference being that we have to deal with a continuous function $\mathcal{U}\to\mathbb{R}$ instead of $\mathbf{U}\to\mathbb{R}$.
\end{proof}

\section{Penalising the Discrete HJB Equation}\label{Section_PenalisationDiscrProb}

We begin by studying Problem \ref{DiscreteProbDef}, which we call the \textit{HJB equation}; for nominal differentiation, Problem \ref{DiscreteProbDef_MinMax}, which we discuss subsequently, will be referred to as the \textit{HJB obstacle problem}.
The penalty approximation used in this paper is an extension of ideas used for discretely controlled HJB equations in \cite{WitteReisinger_PenaltyScheme_DiscreteControlledHJBEquations} and for American options in \cite{ForsythQuadraticConvergence}. The penalised equation is a non-linear equation itself and can be solved by an iterative scheme (cf.\,Section\,\ref{SectionModNewton}).\bigskip

We introduce the penalty term
\begin{equation}
\Pi :\mathbb{R}^N\to\mathbb{R}^N: y\mapsto \big(\pi_1(y_1),\pi_2(y_2),\ldots,\pi_N(y_N)\big)^{tr}.\label{GenPenaltyTermDef}
\end{equation}

\begin{assumption}\label{GenPenaltyTermDef_PropAssumption}
For $i\in\mathcal{N}$, we take $\pi_i\in C(\mathbb{R})$ to be a non-decreasing function
satisfying $\pi_i|_{(-\infty,0]}= 0$ and $\pi_i|_{(0,\infty)}> 0$.
\end{assumption}

\begin{prob}\label{PenDiscreteProbDef}
Let $u_0\in\mathbf{U}$, $\rho>0$. Find $x_{\rho}\in\mathbb{R}^N$ such that
\begin{equation}
(A_{u_0}\,x_{\rho} - b_{u_0}) -\rho \max_{u\in\mathbf{U}}\Pi(b_u-A_u\,x_{\rho}) =0.\label{PenDiscreteProbDef_Eq1}
\end{equation}
\end{prob}

In \eqref{PenDiscreteProbDef_Eq1}, contrary to \cite{WitteReisinger_PenaltyScheme_DiscreteControlledHJBEquations}, a penalty is applied only to the maximum violation of the constraints; since the solution to Problem \ref{PenDiscreteProbDef} can be shown to be bounded independently of $\rho$ (cf.\,Lemma\,\ref{PenDiscreteProbDef_BoundedIndepRho}), this guarantees all constraints to be satisfied as $\rho\to\infty$ (cf.\,Corollary\,\ref{PenaltyConvergenceToTrueSol}).\bigskip

We point out that, in Problem \ref{PenDiscreteProbDef}, $u_0$ does not require any further specification, and
all the results of this section can be shown to hold for any choice of $u_0\in\mathbf{U}$.
In Section \ref{Section_Examples_Indifference}, we will study the  practical effects of $u_0$
in a numerical example.

\begin{la}\label{PenDiscreteProbDef_Uniq}
If there exists a solution to Problem \ref{PenDiscreteProbDef}, then it is unique.
\end{la}
\begin{proof}
Suppose we have two solutions $x_{\rho 1}$ and $x_{\rho 2}$\,. From \eqref{PenDiscreteProbDef_Eq1}, we get that
\begin{equation}
A_{u_0}(x_{\rho 1} - x_{\rho 2}) -\rho \max_{u\in\mathbf{U}}\Pi(b_u-A_u\,x_{\rho 1})
+\rho \max_{u\in\mathbf{U}}\Pi(b_u-A_u\,x_{\rho 2})=0.\label{PenDiscreteProbDef_Uniq_Eq2.8}
\end{equation}
Now, let $i\in\mathcal{N}$, and define $u_{i1}$\,, $u_{i2}\in\mathbf{U}$ to be such that
\begin{equation*}
\max_{u\in\mathbf{U}}\pi_i\big((b_u - A_u\,x_{\rho j})_i\big)=\pi_i\big((b_{u_{ij}} - A_{u_{ij}}\,x_{\rho j})_i\big),\quad j\in\{1,2\}.
\end{equation*}
For the second and third term in \eqref{PenDiscreteProbDef_Uniq_Eq2.8}, we have
\begin{align}
&\ -\rho \max_{u\in\mathbf{U}}\pi_i\big((b_u-A_u\,x_{\rho 1})_i\big)
+\rho \max_{u\in\mathbf{U}}\pi_i\big((b_u-A_u\,x_{\rho 2})_i\big)\label{PenDiscreteProbDef_Uniq_Eq3}\\
=&\ 
-\rho \pi_i\big((b_{u_{i1}}-A_{u_{i1}}\,x_{\rho 1})_i\big)
+\rho \pi_i\big((b_{u_{i2}}-A_{u_{i2}}\,x_{\rho 2})_i\big)\nonumber\\
\leq&\
-\rho \pi_i\big((b_{u_{i2}}-A_{u_{i2}}\,x_{\rho 1})_i\big)
+\rho \pi_i\big((b_{u_{i2}}-A_{u_{i2}}\,x_{\rho 2})_i\big).\nonumber
\end{align}
At this point, the idea is that if there exists a matrix $A^*\in K^o_N$ satisfying
\begin{equation*}
A^*(x_{\rho 1} - x_{\rho 2}) \geq 0,
\end{equation*}
then we can deduce $(x_{\rho 1} - x_{\rho 2}) \geq 0$ since $(A^*)^{-1}\geq 0$, and we get uniqueness as in Theorem \ref{BothDiscreteProblems_UniqSolvable}.
We distinguish between the following two cases.
\begin{itemize}
\item If $\pi_i\big((b_{u_{i2}}-A_{u_{i2}}\,x_{\rho 1})_i\big)\geq \pi_i\big((b_{u_{i2}}-A_{u_{i2}}\,x_{\rho 2})_i\big)\geq 0$, then \eqref{PenDiscreteProbDef_Uniq_Eq3} is non-positive, which means the first term in \eqref{PenDiscreteProbDef_Uniq_Eq2.8} must be non-negative, and we can use $(A^*)_i=(A_{u_0})_i$\,.
\item If $0\leq \pi_i\big((b_{u_{i2}}-A_{u_{i2}}\,x_{\rho 1})_i\big)\leq \pi_i\big((b_{u_{i2}}-A_{u_{i2}}\,x_{\rho 2})_i\big)$, then
$(b_{u_{i2}}-A_{u_{i2}}\,x_{\rho 2})_i\geq (b_{u_{i2}}-A_{u_{i2}}\,x_{\rho 1})_i$\,, which is
equivalent to $\big(A_{u_{i2}}(x_{\rho 1} - x_{\rho 2})\big)_i\geq 0$, and we can set $(A^*)_i = (A_{u_{i2}})_i$\,.
\end{itemize}
\end{proof}

In Section \ref{SectionModNewton}, we will discuss several choices of the penalty function $\Pi$ that allow to conclude the existence of
a solution to Problem \ref{PenDiscreteProbDef}. Here, we study the approximation of Problem \ref{DiscreteProbDef} by Problem \ref{PenDiscreteProbDef}.

\begin{la}\label{PenDiscreteProbDef_BoundedIndepRho}
Suppose there exists a solution $x_{\rho}$ to Problem \ref{PenDiscreteProbDef}.
There exists a constant $C>0$, independent of $\rho$ and $\Pi$, such that $\|x_{\rho}\|_{\infty}\leq C$.
\end{la}
\begin{proof}
We rewrite \eqref{PenDiscreteProbDef_Eq1} to get
\begin{equation}
A_{u_0}\,x_{\rho} =\rho\max_{u\in\mathbf{U}}\Pi\big(b_u-A_u\,x_{\rho}\big)+b_{u_0}.\label{PenDiscreteProbDef_BoundedIndepRho_Eq1}
\end{equation}
Hence, for every component $i\in\mathcal{N}$, we have that
\begin{equation*}
(A_{u_0}\,x_{\rho})_i= (b_{u_0})_i\quad\text{or}\quad\exists\ u\in\mathbf{U}\text{ s.t. }(A_u\,x_{\rho})_i < (b_u)_i\,.
\end{equation*}
Therefore, we may, for $i\in\mathcal{N}$, define $u_i\in\mathbf{U}$ to be such that $(A_{u_i}\,x_{\rho})_i\leq (b_{u_i})_i$
is satisfied, and, introducing $A^*\in K^o_N$ to be the matrix having as $i$-th row the $i$-th row of $A_{u_i}$\,, $i\in\mathcal{N}$, and defining $b^*\in\mathbb{R}^N$ correspondingly, we have
\begin{equation*}
A^*\,x_{\rho}\leq b^*.
\end{equation*}
From Corollary \ref{BoundednessofMatrixandInv_Cor}, we then get $x_{\rho}\leq (A^*)^{-1}b^*\leq C$ for some constant $C$ independent of $\rho$
and $\Pi$.
To see that the negative part of $x_{\rho}$ is bounded as well, we note that, from
\eqref{PenDiscreteProbDef_BoundedIndepRho_Eq1}, we have
\begin{equation*}
x_{\rho} \geq (A_{u_0})^{-1}b_{u_0}\,,
\end{equation*}
in which $(A_{u_0})^{-1}\geq 0$.
\end{proof}

\begin{la}\label{PenaltyApproxAccuracy_Theorem}
Suppose there exists a solution $x_{\rho}$ to Problem \ref{PenDiscreteProbDef} for every $\rho>0$. There exists a constant $C>0$, independent of $\rho$ and $\Pi$, such that
\begin{equation}
\Big\|\min_{u\in\mathbf{U}}\Big[ \max\{A_u\,x_{\rho}-b_u\,,0\} - \Pi(b_u-A_u\,x_{\rho})\Big]\Big\|_{\infty}\leq C/\rho.\label{PenaltyApproxAccuracy_Theorem_Eq.5}
\end{equation}
\end{la}
\begin{proof}
For $u\in\mathbf{U}$, we have
\begin{align*}
\rho\Pi(b_u - A_u\,x_{\rho})\leq&\ \rho \max_{v\in\mathbf{U}}\Pi\big(b_v-A_v\,x_{\rho}\big)\\
=&\ A_{u_0}\,x_{\rho} - b_{u_0}\,,
\end{align*}
and, applying Lemma \ref{PenDiscreteProbDef_BoundedIndepRho} to the last expression, we get that
\begin{equation}
0\leq \Pi(b_u - A_u\, x_{\rho}) \leq C/\rho\label{PenaltyApproxAccuracy_Theorem_Eq1}
\end{equation}
for some constant $C>0$ independent of $\rho$.
Furthermore, for every $i\in\mathcal{N}$, we either have $(A_{u_0}\,x_{\rho}-b_{u_0})_i=0$, in which case
\begin{equation*}
\max\{(A_{u_0}\,x_{\rho}-b_{u_0})_i\,,0\} - \pi_i\big((b_{u_0}-A_{u_0}x_{\rho})_i\big) = 0\leq C/\rho,
\end{equation*}
or $\exists\ u^*\in\mathbf{U}$ such that $\pi_i\big((b_{u^*}-A_{u^*}\,x_{\rho})_i\big) \geq 0$, which, based on \eqref{PenaltyApproxAccuracy_Theorem_Eq1}, gives
\begin{equation*}
\max\{(A_{u^*}\,x_{\rho}-b_{u^*})_i\,,0\} - \pi_i\big((b_{u^*}-A_{u^*}x_{\rho})_i\big)
= - \pi_i\big((b_{u^*}-A_{u^*}x_{\rho})_i\big) \geq -C/\rho.
\end{equation*}
\end{proof}

For the canonical choice $\Pi(\cdot)=\max\{\cdot,0\}$, \eqref{PenaltyApproxAccuracy_Theorem_Eq.5} reduces to
\begin{equation}
\big\|\min_{u\in\mathbf{U}}\{A_u\,x_{\rho}-b_u\}\big\|_{\infty}\leq C/\rho,\label{PenaltyApproxAccuracy_Theorem_Eq2}
\end{equation}
stating that $x_{\rho}$ satisfies Problem \ref{DiscreteProbDef} to an order $O(1/\rho)$; if the set $\mathbf{U}$ is discrete,
estimate \eqref{PenaltyApproxAccuracy_Theorem_Eq2} matches similar results in
\cite{ForsythQuadraticConvergence} and \cite{WitteReisinger_PenaltyScheme_DiscreteControlledHJBEquations}.

\begin{cor}\label{PenaltyConvergenceToTrueSol}
Suppose there exists a solution $x_{\rho}$ to Problem \ref{PenDiscreteProbDef} for every $\rho>0$.
As $\rho\to\infty$, $x_{\rho}$ converges to a limit $x^*$ which solves Problem \ref{DiscreteProbDef}.
\end{cor}
\begin{proof}
Since $(x_{\rho})_{\rho>0}$ is bounded (as seen in Lemma \ref{PenDiscreteProbDef_BoundedIndepRho}), it has a convergent subsequence, which we do not distinguish notationally; we denote the limit of this subsequence by $x^*$.
From Lemma \ref{PenaltyApproxAccuracy_Theorem}, we may deduce that 
\begin{align}
A_u\,x^*-b_u\geq 0,\quad u\in\mathbf{U}.\label{PenaltyConvergenceToTrueSol_Eq.5}
\end{align}
Now, consider a component $i\in\mathcal{N}$.
Using again Lemma \ref{PenaltyApproxAccuracy_Theorem}, we know that, for every $\rho>0$, there exists a $u_{i\rho}\in\mathbf{U}$ such that
\begin{equation}
 0\leq (A_{u_{i\rho}}\,x_{\rho}-b_{u_{i\rho}})_i\leq C/\rho\quad\text{or}\quad 0\leq \pi_i\big((b_{u_{i\rho}}-A_{u_{i\rho}}\,x_{\rho})_i\big)\leq C/\rho.\label{PenaltyConvergenceToTrueSol_Eq1}
\end{equation}
Recalling that $\mathbf{U}$ is compact and $(u_{i\rho})_{\rho>0}\subset\mathbf{U}$, we can then infer that there exists once more a subsequence of $(x_{\rho})_{\rho> 0}$\,,
again not notationally distinguished, and a $u^*_i\in\mathbf{U}$ such that $u_{i\rho}\to u^*_i$\,; hence, recalling the continuity of $\mathfrak{A}$ and $\mathfrak{B}$ in the definition of Problem \ref{DiscreteProbDef}, it follows from \eqref{PenaltyConvergenceToTrueSol_Eq1} that
\begin{equation}
|(A_{u^*_i}\,x^*-b_{u^*_i})_i|=0.\label{PenaltyConvergenceToTrueSol_Eq2}
\end{equation}
Altogether, combining \eqref{PenaltyConvergenceToTrueSol_Eq.5} and \eqref{PenaltyConvergenceToTrueSol_Eq2}, we see that $x^*$ solves Problem \ref{DiscreteProbDef}.
Since Problem \ref{DiscreteProbDef} has a unique solution (cf.\,Theorem\,\ref{BothDiscreteProblems_UniqSolvable}), the just given result holds not only for subsequences of $(x_{\rho})_{\rho>0}$\,, but the whole sequence converges.
\end{proof}

We have finally done enough preparatory work to show that -- if we choose the penalty function $\Pi$ correctly -- the solution to Problem \ref{PenDiscreteProbDef} is indeed a good approximation of the solution to Problem \ref{DiscreteProbDef}; in fact, we can obtain an approximation that is of first order in the penalty parameter.

\begin{theorem}\label{PenaltyConvergenceToTrueSol_ErrorEstimate}
Suppose there exists a solution $x_{\rho}$ to Problem \ref{PenDiscreteProbDef} for every $\rho>0$.
Let $x^*$ be the solution to Problem \ref{DiscreteProbDef}.
If, in addition to Assumption \ref{GenPenaltyTermDef_PropAssumption}, we have $\pi_i|_{(0,\infty)}\in C^1((0,\infty))$, $i\in\mathcal{N}$, and $\inf\{\frac{\partial \pi_i}{\partial y}(y) : y\in(0,\infty),\ i\in\mathcal{N}\}\geq c_{min}>0$, then there exists a constant $C>0$, independent of $\rho$ and $\Pi$ but depending on $c_{min}$\,, such that
\begin{equation*}
\|x^*-x_{\rho}\|_{\infty}\leq C/\rho.
\end{equation*}
\end{theorem}
\begin{proof}
Without loss of generality, we may assume that $c_{min}<1$.
Let $i\in\mathcal{N}$. We use $u^*_i$ as introduced in the proof of Corollary \ref{PenaltyConvergenceToTrueSol}, which means, in particular, that \eqref{PenaltyConvergenceToTrueSol_Eq2} holds. Furthermore, for $\rho>0$, we define
$u^{min}_{i\rho}\in\mathbf{U}$ to be such that
\begin{equation*}
u^{min}_{i\rho} = 
\begin{cases} \argmax_{u\in\mathbf{U}}\pi_i\big((b_u-A_u\,x_{\rho})_i\big) & \text{if $(A_{u_0}x_{\rho}-b_{u_0})_i>0$,}\\
u_0 & \text{if $(A_{u_0}x_{\rho}-b_{u_0})_i=0$,}
\end{cases}
\end{equation*}
which means, as seen in the proof of Lemma \ref{PenaltyApproxAccuracy_Theorem}, that
\begin{align*}
&\ |(A_{u^{min}_{i\rho}}\,x_{\rho}-b_{u^{min}_{i\rho}})_i|\\
\leq&\ \frac{1}{c_{min}}\big|\max\{(A_{u^{min}_{i\rho}}\,x_{\rho}-b_{u^{min}_{i\rho}})_i\,,0\} - \pi_i\big((b_{u^{min}_{i\rho}}-A_{u^{min}_{i\rho}}x_{\rho})_i\big)\big| \leq C_1/\rho
\end{align*}
for some constant $C_1>0$ independent of $\rho$ and $\Pi$. Applying \eqref{PenaltyConvergenceToTrueSol_Eq1}, \eqref{PenaltyConvergenceToTrueSol_Eq2} and the definition of $u^{min}_{i\rho}$, we now have that
\begin{align*}
(A_{u^{min}_{i\rho}})_i\,(x_{\rho}-x^*)\leq (A_{u^{min}_{i\rho}}\, x_{\rho}-b_{u^{min}_{i\rho}})_i
\leq C_1/\rho
\end{align*}
and
\begin{align*}
(A_{u^*_i})_i\,(x^*-x_{\rho})\leq&\
(A_{u^*_i}\,x^*-b_{u^*_i})_i-
(A_{u^{min}_{i\rho}}\, x_{\rho}-b_{u^{min}_{i\rho}})_i\\
=&\ -(A_{u^{min}_{i\rho}}\, x_{\rho}-b_{u^{min}_{i\rho}})_i
\leq C_1/\rho.
\end{align*}
Denoting by $A_1^*$\,, $A_2^*\in\ K^o_N$ the matrices having as $i$-th row, $i\in\mathcal{N}$, the $i$-th rows of $A_{u^{min}_{i\rho}}$ and $A_{u^*_i}$\,, respectively, we get that
\begin{equation*}
x_{\rho}-x^*\leq \frac{\|(A^*_1)^{-1}\|_{\infty}\,C_1}{\rho}\quad\text{and}\quad x^*-x_{\rho}\leq \frac{\|(A^*_2)^{-1}\|_{\infty}\,C_1}{\rho}\,,
\end{equation*}
and, using Corollary \ref{BoundednessofMatrixandInv_Cor}, we may infer that
\begin{equation*}
\|x^*-x_{\rho}\|_{\infty}\leq C_2/\rho
\end{equation*}
for some constant $C_2>0$ independent of $\rho$ and $\Pi$.
\end{proof}

\section{Solving the Penalised HJB Equation by Iteration}\label{SectionModNewton}

In the previous section, we have seen how Problem \ref{DiscreteProbDef} can be approximated by penalisation.
We will now discuss iterative methods for the solution of the penalised problem, i.e. algorithms for the computation
of $y\in\mathbb{R}^N$ satisfying \eqref{PenDiscreteProbDef_Eq1}.\bigskip

Defining
\begin{equation*}
G(y):=
(A_{u_0}\,y - b_{u_0}) -\rho \max_{u\in\mathbf{U}}\Pi\big(b_u-A_u\,y\big),\quad y\in\mathbb{R}^N,
\end{equation*}
we need to solve $G(y)=0$.
For $i\in\mathcal{N}$, $y\in\mathbb{R}^N$, we define
\begin{equation*}
u^{min}_{\pi_i}(y) := \argmax_{v\in\mathbf{U}}\pi_i\big((b_v-A_v\,y)_i\big),
\end{equation*}
and we assume that $\pi_i|_{(0,\infty)}\in C^1((0,\infty))$.
For some of the following theorems, we require further that, for $i$, $j\in\mathcal{N}$,
\begin{equation}
\frac{\partial }{\partial y_j}\Big[\max_{u\in\mathbf{U}}\pi_i\big((b_u-A_u\,y)_i\big)\Big] =
\frac{\partial \pi_i}{\partial y}\big((b_{u^{min}_{\pi_i}(y)}-A_{u^{min}_{\pi_i}(y)}\,y)_i\big)\big(A_{u^{min}_{\pi_i}(y)}\big)_{ij}\label{SectionModNewton_Eq0.95}
\end{equation}
holds whenever $\frac{\partial \pi_i}{\partial y}$ is well defined, with $u^{min}_{\pi_i}\in C(\mathbb{R}^N,\mathbf{U})$; the Jacobian of $G$ is then given by
\begin{equation}
\big(J_G(y)\big)_i=(A_{u_0})_{i}\ +\rho\frac{\partial \pi_i}{\partial y}\big((b_{u^{min}_{\pi_i}(y)}-A_{u^{min}_{\pi_i}(y)}\,y)_i\big)\big(A_{u^{min}_{\pi_i}(y)}\big)_i\label{SectionModNewton_Eq1}
\end{equation}
for $i\in\mathcal{N}$, $y\in\mathbb{R}^N$; if $\frac{\partial \pi_i}{\partial y}$ does not exist at 0, we set $\frac{\partial \pi_i}{\partial y}(0):=\lim_{y\uparrow 0}\frac{\partial \pi_i}{\partial y}(y) = 0$.
For the numerical examples following later in this paper, \eqref{SectionModNewton_Eq0.95} is generally satisfied.\bigskip


The next theorem states that a globally convergent iterative scheme for the solution of Problem \ref{PenDiscreteProbDef} exists
whenever $\Pi$ is a smooth and non-decreasing function.

\begin{theorem}\label{LineSearchNewton_HJB}
If, in addition to Assumption \ref{GenPenaltyTermDef_PropAssumption}, for $i\in\mathcal{N}$, $\pi_i\in C^1(\mathbb{R})$ and $\frac{\partial \pi_i}{\partial y}\geq 0$, and \eqref{SectionModNewton_Eq0.95} is satisfied, then Newton's method with line search as introduced by J.-S. Pang in \cite{Pang_NewtonsMethod} presents a globally convergent iterative scheme for the solution of Problem \ref{PenDiscreteProbDef}. In particular, this means that there exists a solution to Problem \ref{PenDiscreteProbDef}.
\end{theorem}
\begin{proof}
We will show that conditions $(a)$-$(d)$ in Theorem 4 in \cite{Pang_NewtonsMethod} are satisfied.
First of all, since $G$ is continuously partially differentiable, it is differentiable as a function $\mathbb{R}^N\to\mathbb{R}^N$, which, in particular, guarantees $B$-differentiability of $G$.
In $(a)$, we need to show that, for arbitrary $y^0\in\mathbb{R}^N$, the set
$\{y\in\mathbb{R}^N : \|G(y)\|_{\infty}\leq \|G(y^0)\|_{\infty} \}$ is bounded; since, for $i\in\mathcal{N}$, we have $(A_{u_0}\,y)_i\geq (b_{u_0})_i$\,, and either  $(A_{u_0}\,y)_i\leq \big(G(y^0)+b_{u_0}\big)_i$
or $(A_{u^{min}_{\pi_i}(y)}\,y)_i \leq (b_{u^{min}_{\pi_i}(y)})_i$\,, this
can easily be inferred by applying Corollary \ref{BoundednessofMatrixandInv_Cor}.
The partial derivatives of $\Pi$ are non-negative, which means $J_G(y)\in K^o_N$\,, $y\in\mathbb{R}^N$, and, hence, $(b)$ holds.
Given the continuity of the partial derivatives of $G$, we can apply Lemma 1 and Theorem 2 in \cite{Pang_NewtonsMethod} to obtain $(c)$. Finally, in $(d)$, it is
sufficient to show that, for any compact set $\mathcal{B}\subset\mathbb{R}$, there exists a constant $c>0$ such that
\begin{equation}
\|J_G(y)v\|_{\infty}\geq c\label{LineSearchNewton_HJB_Eq1}
\end{equation}
for all $y\in\mathcal{B}$, $v\in\mathbb{R}^N$, $\|v\|_{\infty}=1$. We prove \eqref{LineSearchNewton_HJB_Eq1} by contradiction, i.e. suppose there exist sequences
$(y^n)_{n\geq 1}\subset\mathcal{B}$ and $(v^n)_{n\geq 1}\subset\mathbb{R}^N$, $\|v^n\|_{\infty}=1$ for all $n\geq 1$, such that $\lim_{n\to\infty}\|J_G(y^n)v^n\|_{\infty}=0$. Based on the boundedness of the sequences $(y^n)_{n\geq 1}$ and $(v^n)_{n\geq 1}$\,, we infer the existence of subsequences -- not notationally distinguished -- converging to limits $y^*\in\mathcal{B}$ and $v^*\in\mathbb{R}^N$, respectively, where $\|v^*\|_{\infty}=1$, and we get $\lim_{n\to\infty}\|J_G(y^n)v^n\|_{\infty}=\|J_G(y^*)v^*\|_{\infty}=0$.
Now, from $\frac{\partial \pi_i}{\partial y}\geq 0$, $i\in\mathcal{N}$, we have that $J_G(y^*)\in K_N^o$\,, and we note that $J_G(y^*)v^*=0$ implies $v^*=0$, which is a contradiction to $\|v^*\|_{\infty}=1$. Hence, we may conclude that \eqref{LineSearchNewton_HJB_Eq1} must hold true, which completes the proof.
\end{proof}

Having seen that Newton's method with line search can lead to a globally convergent scheme, we next come to look at
Newton's method in its classical form.

\begin{alg}\label{NewtonAlg_HJB}(Newton's Method for the pen.\ HJB Equation)
Let $x^0\in\mathbb{R}^N$ be some starting value. Then, for known $x^n$, $n\geq 0$, find $x^{n+1}$ such that
\begin{equation}
J_G(x^n)(x^{n+1}-x^n)  = -G(x^n).\label{NewtonAlg_HJB_Eq1}
\end{equation}
\end{alg}

\begin{theorem}\label{NewtonAlg_HJB_PropTheorem}
Suppose the conditions of Theorem \ref{LineSearchNewton_HJB} are satisfied, and let $x_{\rho}$ be the solution to Problem \ref{PenDiscreteProbDef}. There exists a neighbourhood $\mathcal{B}$ of $x_{\rho}$ such that, for any starting value $x^0\in\mathcal{B}$, the Newton sequence $(x^n)_{n\geq 0}$ generated by Algorithm \ref{NewtonAlg_HJB} is well defined, remains in $\mathcal{B}$ and converges to $x_{\rho}$\,; the rate of convergence is quadratic.
\end{theorem}
\begin{proof}
If $G$ has a strong and non-singular $F$-derivative at $x_{\rho}$, then the result can be found in Theorem 3 in \cite{Pang_NewtonsMethod}. Given the continuity of the partial derivatives of $G$, we can apply Theorem 2 in \cite{Pang_NewtonsMethod} to obtain the existence of a strong $F$-derivative of $G$, which
coincides with $J_G$. Since $\frac{\partial \pi_i}{\partial y}\geq 0$, $i\in\mathcal{N}$, implies $J_G(x_{\rho})\in K_N^o$\,, the strong $F$-derivative of $G$ at $x_{\rho}$ is indeed non-singular.
\end{proof}

Now, if we use $\Pi(\cdot)=\max\{\cdot,0\}$ (in which case, based on \eqref{SectionModNewton_Eq1}, $J_G$ is still well defined), \eqref{NewtonAlg_HJB_Eq1} simplifies, and we get the following special case of Algorithm \ref{NewtonAlg_HJB}.

\begin{alg}\label{ModifiedNewtonAlg} (Newton-like Method for pen.\ HJB Eq.)
For $i\in\mathcal{N}$ and $y\in\mathbb{R}^N$, define
$u^{min}_i(y)\in\mathbf{U}$ to be such that
\begin{equation}
u^{min}_i(y) = \argmax_{v\in\mathbf{U}}\Big[\max\{(b_v-A_v\,y)_i\,,0\}\Big],\label{ModifiedNewtonAlg_Eq0.5}
\end{equation}
and set $A^{min}(y)\in\mathbb{R}^{N\times N}$ and $b^{min}(y)\in\mathbb{R}^N$ to be matrix and vector consisting of
\begin{itemize}
\item rows $(A_{u^{min}_i(y)})_i$ and $(b_{u^{min}_i(y)})_i$\,, $i\in\mathcal{N}$, respectively, if we have $(b_{u^{min}_i(y)}-A_{u^{min}_i(y)}\,y)_i>0$
\item and having zero rows if $(b_{u^{min}_i(y)}-A_{u^{min}_i(y)}\,y)_i\leq0$.
\end{itemize} 
If $\Pi(y)=\max\{y,0\}$, we now have $J_G(y)=A_{u_0} +  \rho A^{min}(y)\in K^o_N$\,, and \eqref{NewtonAlg_HJB_Eq1} becomes
\begin{multline*}
\big(A_{u_0} + \rho A^{min}(x^n)\big)(x^{n+1}-x^n)
=
-(A_{u_0}\,x^n - b_{u_0})+
b^{min}(x^n) -\rho A^{min}(x^n)x^n\,,
\end{multline*}
which is equivalent to
\begin{align}
\big(A_{u_0} + \rho A^{min}(x^n)\big)x^{n+1}
= b_{u_0} +\rho\,b^{min}(x^n).\label{ModifiedNewtonAlg_Eq2}
\end{align}
\end{alg}

\begin{la}\label{NewtonAlgMonotonicity}
Let $x^0\in\mathbb{R}^N$ be some starting value, and let $(x^n)^{\infty}_{n=0}$ be the sequence generated
by Algorithm \ref{ModifiedNewtonAlg}. We have $x^n\leq x^{n+1}$ for $n\geq 1$.
\end{la}
\begin{proof}
Writing $\eqref{ModifiedNewtonAlg_Eq2}$ for $x^n$ and $x^{n+1}$, we obtain
\begin{align*}
\big(A_{u_0} + \rho\,A^{min}(x^n)\big)x^{n+1}
=&\ b_{u_0} +\rho\,b^{min}(x^n)\\
\text{and}\quad
\big(A_{u_0} + \rho\,A^{min}(x^{n-1})\big)x^n
=&\ b_{u_0} +\rho\,b^{min}(x^{n-1}),
\end{align*}
and subtracting yields
\begin{multline}
\big(A_{u_0} + \rho\,A^{min}(x^n)\big)(x^{n+1}-x^n)\\
=
\rho\,b^{min}(x^n)
- \rho\,A^{min}(x^n)\,x^n
-\rho\,b^{min}(x^{n-1})
+ \rho\,A^{min}(x^{n-1})\,x^n\,.\label{NewtonAlgMonotonicity_Eq1}
\end{multline}
In this, $A_{s_0} + \rho\,A^{min}(x^n)\in K^o_N$ has a non-negative inverse, and, therefore, the proof is complete if we can show that the right-hand side of expression \eqref{NewtonAlgMonotonicity_Eq1} is non-negative; to do this, we consider the rows $i\in\mathcal{N}$ of the right-hand side of \eqref{NewtonAlgMonotonicity_Eq1} separately.
\begin{itemize}
\item[(i)] If $\big(b^{min}(x^{n-1})
-A^{min}(x^{n-1})\,x^{n-1}\big)_i=0$, the right-hand side of \eqref{NewtonAlgMonotonicity_Eq1} equals
$\big(\rho\,b^{min}(x^n)
- \rho\,A^{min}(x^n)\,x^n\big)_i\geq 0$.
\item[(ii)] If $\big(b^{min}(x^{n-1})
-A^{min}(x^{n-1})\,x^{n-1}\big)_i>0$, the right-hand side of \eqref{NewtonAlgMonotonicity_Eq1} equals
\begin{align*}
&\ \big(\rho\,b^{min}(x^n)
- \rho\,A^{min}(x^n)\,x^n
-\rho\,b^{min}(x^{n-1})
+ \rho\,A^{min}(x^{n-1})\,x^n\big)_i\\
\geq&\ \big(\rho\,b^{min}(x^n)
- \rho\,A^{min}(x^n)\,x^n
-\rho\,b^{min}(x^n)
+ \rho\,A^{min}(x^n)\,x^n\big)_i=0.
\end{align*}
\end{itemize}
This completes the proof.
\end{proof}

\begin{la}\label{NewtonAlg_FiniteTermination}
Let $(x^n)^{\infty}_{n=0}$ be the sequence generated
by Algorithm \ref{ModifiedNewtonAlg}. There exists a constant $C>0$ such that, for every starting value $x^0$\,, $\|x^n\|_{\infty}\leq C$, $n\geq 1$.
\end{la}
\begin{proof}
We may write $\eqref{ModifiedNewtonAlg_Eq2}$ as
\begin{align}
x^{n+1}
= \big(A_{u_0} + \rho A^{min}(x^n)\big)^{-1}\,\big(b_{u_0} +\rho\,b^{min}(x^n)\big).\label{ModifiedNewtonAlg_Eq3}
\end{align}
The boundedness of the right-hand side of \eqref{ModifiedNewtonAlg_Eq3} follows from Corollary \ref{BoundednessofMatrixandInv_Cor}.
\end{proof}

\begin{theorem}\label{NewtonLimitSolvesPenProblem}
Independently of the starting value $x^0$\,, Algorithm \ref{ModifiedNewtonAlg} converges to a limit $x_{\rho}$ which solves Problem \ref{PenDiscreteProbDef}.
\end{theorem}
\begin{proof}
From Lemmas \ref{NewtonAlgMonotonicity} and \ref{NewtonAlg_FiniteTermination}, we have the existence of a limit $x_{\rho}$ such that $x^n\to x_{\rho}$ as $n\to\infty$. Hence, it only remains to show that the limit $x_{\rho}$ satisfies $G(x_{\rho})=0$. This follows easily from \eqref{NewtonAlg_HJB_Eq1}, the boundedness of $J_G$ and the continuity of the function $G$.
\end{proof}

\section{Penalising the Discrete Obstacle Problem}\label{Penalise_Discr_Obst_prob}

Having dealt with Problem \ref{DiscreteProbDef} (termed \textit{HJB equation}) in the previous 	sections, we now
consider Problem \ref{DiscreteProbDef_MinMax} (termed \textit{HJB obstacle problem}), which has to be treated differently due
to the combination of ``min'' and ``max'' operators. We use the same penalty term $\Pi$ as introduced in \eqref{GenPenaltyTermDef}, and we again make Assumption \ref{GenPenaltyTermDef_PropAssumption}.

\begin{prob}\label{PenDiscreteProbDef_MinMax}
Let $\rho>0$. Find $z_{\rho}\in\mathbb{R}^N$ such that
\begin{equation}
\max_{u\in\mathbf{U}}\{A_u\,z_{\rho}-b_u\}-\rho\Pi\big(\tilde{b}-\tilde{A}z_{\rho}\big) =0.\label{PenDiscreteProbDef_MinMax_Eq1}
\end{equation}
\end{prob}

If one treats the ``max'' inside the ``min'' in \eqref{DiscreteProbDef_MinMax_Eq1} as if $|\mathbf{U}|=1$, then
the penalty formulation \eqref{PenDiscreteProbDef_MinMax_Eq1} corresponds to the penalisation technique
introduced in \cite{ForsythQuadraticConvergence}, and, thus, \eqref{PenDiscreteProbDef_MinMax_Eq1} has a heuristic interpretation; however, mathematically, the situation is a bit more subtle, and we will now study various properties of Problem \ref{PenDiscreteProbDef_MinMax}.

\begin{la}\label{PenDiscreteProbDef_MinMax_Uniq}
If there exists a solution to Problem \ref{PenDiscreteProbDef_MinMax}, then it is unique.
\end{la}
\begin{proof}
Suppose we have two solutions $z_{\rho 1}$ and $z_{\rho 2}$\,.
For $i\in\mathcal{N}$, we use $u_i^1$\,, $u_i^2\in\mathbf{U}$ as defined in \eqref{DiscreteProb_MinMax_UniqTh_Eq0.5}. From equation \eqref{PenDiscreteProbDef_MinMax_Eq1}, we then get that, for $i\in\mathcal{N}$,
\begin{multline*}
\big(A_{u^1_i}(z_{\rho 1} - z_{\rho 2})\big)_i - \rho \pi_i\big((\tilde{b}- \tilde{A}z_{\rho 1})_i\big)
+\rho \pi_i\big((\tilde{b}- \tilde{A}z_{\rho 2})_i\big)\\
\geq
(A_{u^1_i}z_{\rho 1} -b _{u^1_i} - A_{u^2_i}z_{\rho 2} + b _{u^2_i})_i
 -\rho \pi_i\big((\tilde{b}- \tilde{A}z_{\rho 1})_i\big)
+\rho \pi_i\big((\tilde{b}- \tilde{A}z_{\rho 2})_i\big)
=0.
\end{multline*}
At this point, we have an expression similar to \eqref{PenDiscreteProbDef_Uniq_Eq2.8}, and we can finish off by following the lines of the proof of Lemma \ref{PenDiscreteProbDef_Uniq}.
\end{proof}

\begin{la}\label{PenDiscreteProbDef_MinMax_BoundednessINdepOfRhoandPi}
Assume there exists a solution $z_{\rho}$ to Problem \ref{PenDiscreteProbDef_MinMax} for every $\rho>0$.
There exists a constant $C>0$, independent of $\rho$ and $\Pi$, such that
\begin{equation*}
\|(z_{\rho})_{\rho>0}\|_{\infty}\leq C.
\end{equation*}
\end{la}
\begin{proof}
Applying Corollary \ref{BoundednessofMatrixandInv_Cor}, we can argue as in the proof of Lemma \ref{PenDiscreteProbDef_BoundedIndepRho}.
\end{proof}

\begin{la}\label{PenDiscreteProbDef_MinMax_AllKindsofRhoConvergenceResults}
Assume there exists a solution $z_{\rho}$ to Problem \ref{PenDiscreteProbDef_MinMax} for every $\rho>0$.
There exists a constant $C>0$, independent of $\rho$ and $\Pi$, such that
\begin{equation}
\Big\|\min\Big\{\max_{u\in\mathbf{U}}\{A_u\,z_{\rho}-b_u\},\max\{\tilde{A}z_{\rho}-\tilde{b},0\}-\Pi(\tilde{b}-\tilde{A}z_{\rho})\Big\}\Big\|_{\infty}\leq C/\rho.\label{PenDiscreteProbDef_MinMax_AllKindsofRhoConvergenceResults_Eq.5}
\end{equation}
\end{la}
\begin{proof}
Based on Lemma \ref{PenDiscreteProbDef_MinMax_BoundednessINdepOfRhoandPi}, we can directly infer from \eqref{PenDiscreteProbDef_MinMax_Eq1} that
\begin{equation*}
0\leq \Pi\big(\tilde{b}-\tilde{A}z_{\rho}\big) \leq C/\rho.
\end{equation*}
At this stage, we obtain estimate \eqref{PenDiscreteProbDef_MinMax_AllKindsofRhoConvergenceResults_Eq.5} by pointing out that, for $i\in\mathcal{N}$, either $\max_{u\in\mathbf{U}}(A_u\,z_{\rho}-b_u)_i=0$ (which also means $\pi_i\big((\tilde{b}-\tilde{A}z_{\rho})_i\big)=0$) or $\max\{(\tilde{A}z_{\rho}-\tilde{b})_i\,,0\}-\pi_i\big((\tilde{b}-\tilde{A}z_{\rho})_i\big) \geq -C/\rho$.
\end{proof}

For the canonical choice $\Pi(\cdot)=\max\{\cdot,0\}$, \eqref{PenDiscreteProbDef_MinMax_AllKindsofRhoConvergenceResults_Eq.5} becomes
\begin{equation*}
\Big\|\min\Big\{\max_{u\in\mathbf{U}}\{A_u\,z_{\rho}-b_u\},\tilde{A}z_{\rho}-\tilde{b}\Big\}\Big\|_{\infty}\leq C/\rho,
\end{equation*}
meaning that $z_{\rho}$ satisfies Problem \ref{DiscreteProbDef_MinMax} to an order $O(1/\rho)$; this property is consistent
with the penalty approximation of the HJB equation discussed in Section \ref{Section_PenalisationDiscrProb}.

\begin{cor}\label{PenaltyConvergenceToTrueSol_MinMax}
Suppose there exists a solution $z_{\rho}$ to Problem \ref{PenDiscreteProbDef_MinMax} for every $\rho>0$.
As $\rho\to\infty$, $z_{\rho}$ converges to a limit $z^*$ which solves Problem \ref{DiscreteProbDef_MinMax}.
\end{cor}
\begin{proof}
This can be shown by arguing as in the proof of Corollary \ref{PenaltyConvergenceToTrueSol}.
\end{proof}

\begin{theorem}\label{PenaltyConvergenceToTrueSol_ErrorEstimate_Isaacs}
Suppose there exists a solution $z_{\rho}$ to Problem \ref{PenDiscreteProbDef_MinMax} for every $\rho>0$.
Let $z^*$ be the solution to Problem \ref{DiscreteProbDef_MinMax}.
If, in addition to Assumption \ref{GenPenaltyTermDef_PropAssumption}, we have $\pi_i\in C^1((0,\infty))$, $i\in\mathcal{N}$, and $\inf\{\frac{\partial \pi_i}{\partial y}(y) : y\in(0,\infty),\ i\in\mathcal{N}\}\geq c_{min}>0$, then there exists a constant $C>0$, independent of $\rho$ and $\Pi$ but depending on $c_{min}$\,, such that
\begin{equation*}
\|z^*-z_{\rho}\|_{\infty}\leq C/\rho.
\end{equation*}
\end{theorem}
\begin{proof}
We proceed similarly to the proof of Theorem \ref{PenaltyConvergenceToTrueSol_ErrorEstimate}.
Without loss of generality, we may assume that $c_{min}<1$.
Let $i\in\mathcal{N}$. Let $(A_i^*\,,b^*_i)\in\{(\tilde{A},\tilde{b})\}\cup\{(A_u\,,b_u) : u\in\mathbf{U}\}$ be such that 
\begin{equation*}
\min\Big\{ \max_{u\in\mathbf{U}}(A_u\,z^*-b_u)_i\,,(\tilde{A}z^*-\tilde{b})_i\,\Big\}=(A^*_i\,z^*-b^*_i)_i=0.
\end{equation*}
Similarly, applying Lemma \ref{PenDiscreteProbDef_MinMax_AllKindsofRhoConvergenceResults}, let $(A_{i\rho}\,,b_{i\rho})\in\{(\tilde{A},\tilde{b})\}\cup\{(A_u\,,b_u) : u\in\mathbf{U}\}$ be such that 
\begin{multline*}
|(A_{i\rho}\,z_{\rho}-b_{i\rho})_i|
=\Big|\min\Big\{ \max_{u\in\mathbf{U}}(A_u\,z_{\rho}-b_u)_i\,,(\tilde{A}z_{\rho}-\tilde{b})_i\,\Big\}\Big|\\
\leq \frac{1}{c_{min}}\Big|\min\Big\{\max_{u\in\mathbf{U}}(A_u\,z_{\rho}-b_u)_i\,,\max\{(\tilde{A}z_{\rho}-\tilde{b})_i\,,0\}-\pi_i\big((\tilde{b}-\tilde{A}z_{\rho})_i\big)\Big\}\Big|\\
\leq \frac{C_1}{c_{min}\,\rho}
\end{multline*}
for some constant $C_1>0$ independent of $\rho$ and $\Pi$.
Furthermore, as already in \eqref{DiscreteProb_MinMax_UniqTh_Eq0.5}, let $u_i^*$\,, $u^{\rho}_i\in\mathbf{U}$ such that $(A_{u^*_i}\,z^*-b_{u^*_i})_i = \max_{u\in\mathbf{U}}(A_u\,z^*-b_u)_i$ and $(A_{u^{\rho}_i}\,z_{\rho}-b_{u^{\rho}_i})_i = \max_{u\in\mathbf{U}}(A_u\,z_{\rho}-b_u)_i$\,. We first consider $z_{\rho}-z^*$, for which we distinguish two cases. ($C_2$\,,\,$C_3>0$ are taken to be constants independent of $\rho$ and $\Pi$.)
\begin{itemize}
\item[(i)] If $A_{i\rho}=\tilde{A}$, we have $\big(A_{i\rho}(z_{\rho}-z^*)\big)_i\leq (\tilde{A}z_{\rho}-\tilde{b})_i-(A^*_i\,z^*-b^*_i)_i\leq C_2/\rho$.
\item[(ii)] If $A_{i\rho}=A_{u^{\rho}_i}$\,, we have $\big(A_{u^*_i}(z_{\rho}-z^*)\big)_i\leq (A_{u^{\rho}_i}\,z_{\rho}-b_{u^{\rho}_i})_i-(A^*_i\,z^*-b^*_i)_i\leq C_3/\rho$.
\end{itemize}
Now, for the reverse case, $z^*-z_{\rho}$\,, similar estimates can be obtained by swapping the roles of ``$*$'' and ``$\rho$''.
The proof can be completed by following the lines of the proof of Theorem \ref{PenaltyConvergenceToTrueSol_ErrorEstimate}.
\end{proof}

\section{Solving the Penalised HJB Obstacle Problem by Iteration}\label{SectionModNewton_Isaacs}

In Section \ref{SectionModNewton}, we have discussed how to iteratively solve the penalty approximation of the discrete HJB equation.
Similarly, in this section, we discuss iterative methods for the penalty approximation of the discrete HJB obstacle problem.\bigskip

Defining
\begin{equation*}
H(y):=
\max_{u\in\mathbf{U}}\{A_u\,y-b_u\}-\rho\Pi\big(\tilde{b}-\tilde{A}y,0\big) = 0
\end{equation*}
for $y\in\mathbb{R}^N$, to solve \eqref{PenDiscreteProbDef_MinMax_Eq1}, we need to compute $y$ such that $H(y)=0$.
For $i\in\mathcal{N}$, $y\in\mathbb{R}^N$, we define
\begin{equation}
u^{max}_i(y) := \argmax_{v\in\mathbf{U}}\,(A_v\,y-b_v)_i\,,\label{SectionModNewton_Isaacs_Eq0.9}
\end{equation}
and we assume that $\pi_i|_{(0,\infty)}\in C^1((0,\infty))$, $u^{max}_i\in C(\mathbb{R}^{N},\mathbf{U})$, and
\begin{equation}
\frac{\partial }{\partial y_j}\Big[\max_{u\in\mathbf{U}}(A_u\,y-b_u)_i\Big] = \big(A_{u^{max}_i(y)}\big)_{ij}\in C(\mathbb{R}^N,\mathbb{R})\label{SectionModNewton_Isaacs_Eq1}
\end{equation}
for $i$, $j\in\mathcal{N}$. The Jacobian of $H$ is then given by
\begin{equation}
\big(J_H(y)\big)_i:=\big(A_{u^{max}_i(y)}\big)_{i}\ +\rho\frac{\partial \pi_i}{\partial y}\big((\tilde{b}-\tilde{A}y)_i\big)(\tilde{A})_i\label{SectionModNewton_Isaacs_Eq2}                                                                                                                                           
\end{equation}for $i\in\mathcal{N}$, $y\in\mathbb{R}^N$, whenever $\frac{\partial \pi_i}{\partial y}$ is well defined; if $\frac{\partial \pi_i}{\partial y}$ does not exist at 0, we set $\frac{\partial \pi_i}{\partial y}(0):=\lim_{y\uparrow 0}\frac{\partial \pi_i}{\partial y}(y) = 0$. Like in Section \ref{SectionModNewton}, we point out that, for the numerical examples following later in this paper, \eqref{SectionModNewton_Isaacs_Eq1} is generally satisfied.\bigskip


As already in Theorem \ref{LineSearchNewton_HJB}, whenever $\Pi$ is a smooth and non-decreasing function, we can find a globally convergent iterative scheme for the solution of Problem \ref{PenDiscreteProbDef_MinMax}.

\begin{theorem}\label{LineSearchNewton}
If, in addition to Assumption \ref{GenPenaltyTermDef_PropAssumption}, for $i\in\mathcal{N}$, $\pi_i\in C^1(\mathbb{R})$ and $\frac{\partial \pi_i}{\partial y}\geq 0$, then Newton's method with line search as introduced by J.-S. Pang in \cite{Pang_NewtonsMethod} presents a globally convergent iterative scheme for the solution of Problem \ref{PenDiscreteProbDef_MinMax}. In particular, this means that there exists a solution to Problem \ref{PenDiscreteProbDef_MinMax}.
\end{theorem}
\begin{proof}
The proof is a minor modification of the proof of Theorem \ref{LineSearchNewton_HJB}.
\end{proof}

Knowing that Problem \ref{PenDiscreteProbDef_MinMax} has a solution for sufficiently smooth penalty terms $\Pi$, we can now proceed to showing that there also exists a solution if we penalise using the $max$-function.

\begin{cor}\label{ExistenceOfSolution_Cor_PiEqualsMaxFun}
If $\Pi(\cdot) = \max\{\cdot,0\}$, then there exists
a solution to Problem \ref{PenDiscreteProbDef_MinMax}.
\end{cor}
\begin{proof} Let $(\Pi^{\epsilon})_{\epsilon>0}$ be a sequence of penalty functions satisfying $\pi^{\epsilon}_i\in C^{\infty}(\mathbb{R})$, $0\leq \frac{\partial \pi^{\epsilon}_i}{\partial y}\leq 1$ and $\|\max\{y,0\}-\Pi^{\epsilon}(y)\|_{\infty}\leq \epsilon$ for $i\in\mathcal{N}$, $y\in\mathbb{R}$, $\epsilon>0$, and let $(z^*_{\epsilon})_{\epsilon>0}$ be the sequence of solutions to Problem \ref{PenDiscreteProbDef_MinMax} corresponding
to $(\Pi^{\epsilon})_{\epsilon>0}$\,. Based on Lemma \ref{PenDiscreteProbDef_MinMax_BoundednessINdepOfRhoandPi}, we know that
there exists a subsequence of $(z^*_{\epsilon})_{\epsilon>0}$\,, not notationally distinguished, that converges to a limit $z^*\in\mathbb{R}^N$ as $\epsilon\to 0$. We aim to show that $z^*$ solves Problem \ref{PenDiscreteProbDef_MinMax} for $\Pi(\cdot) = \max\{\cdot,0\}$. We have
\begin{align*}
& \ \|\max_{u\in\mathbf{U}}\{A_u\,z^*-b_u\}-\rho\max\{\tilde{b}-\tilde{A}z^*,0\}\|_{\infty}\\
=&\ \|\max_{u\in\mathbf{U}}\{A_u\,z_{\epsilon}-b_u\}-\rho\Pi^{\epsilon}\big(\tilde{b}-\tilde{A}z_{\epsilon}\big) - \max_{u\in\mathbf{U}}\{A_u\,z^*-b_u\}+\rho\max\{\tilde{b}-\tilde{A}z^*,0\}\|_{\infty}\\
\leq &\ \|\max_{u\in\mathbf{U}}\{A_u\,z_{\epsilon}-b_u\} - \max_{u\in\mathbf{U}}\{A_u\,z^*-b_u\}\|_{\infty}\\
&\quad\quad\quad\quad\quad\quad\quad\quad\quad\quad\quad\quad+\|\rho\Pi^{\epsilon}\big(\tilde{b}-\tilde{A}z_{\epsilon}\big) -\rho\max\{\tilde{b}-\tilde{A}z^*,0\}\|_{\infty}\,,
\end{align*}
in which we obtain
\begin{equation*}
\lim_{\epsilon\to 0} \|\max_{u\in\mathbf{U}}\{A_u\,z_{\epsilon}-b_u\} - \max_{u\in\mathbf{U}}\{A_u\,z^*-b_u\}\|_{\infty} = 0
\end{equation*}
if we argue as in the proof of Corollary \ref{PenaltyConvergenceToTrueSol}, and
\begin{align*}
&\ \|\rho\Pi^{\epsilon}\big(\tilde{b}-\tilde{A}z_{\epsilon}\big) -\rho\max\{\tilde{b}-\tilde{A}z^*,0\}\|_{\infty}\\
\leq&\
\|\rho\Pi^{\epsilon}\big(\tilde{b}-\tilde{A}z^*\big) -\rho\max\{\tilde{b}-\tilde{A}z^*,0\}\|_{\infty}
+ \|\rho\Pi^{\epsilon}\big(\tilde{b}-\tilde{A}z^*\big) -\rho\Pi^{\epsilon}\big(\tilde{b}-\tilde{A}z_{\epsilon}\big)\|_{\infty}\\
\leq &\ \rho\epsilon\ + \rho\|z_{\epsilon}-z^*\|_{\infty}
\end{align*}
as $\epsilon\to 0$ by the properties of $(\Pi^{\epsilon})_{\epsilon>0}$\,; hence, we have $\max_{u\in\mathbf{U}}\{A_u\,z^*-b_u\}-\rho\max\{\tilde{b}-\tilde{A}z^*,0\}=0$.
\end{proof}

Having established that Newton's method with line search can be used to solve the penalised HJB obstacle problem, we proceed
to Newton's method in its classical form.

\begin{alg}\label{NewtonAlg_Isaacs}(Newton's Method for the pen.\ HJB Obstacle Prob.)
Let $z^0\in\mathbb{R}^N$ be some starting value. Then, for known $z^n$, $n\geq 0$, find $z^{n+1}$ such that
\begin{equation}
J_H(z^n)(z^{n+1}-z^n)  = -H(z^n).\label{NewtonAlg_Isaacs_Eq1}
\end{equation}
\end{alg}

\begin{theorem}\label{NewtonAlg_Isaacs_PropTheorem}
Suppose the conditions of Theorem \ref{LineSearchNewton} are satisfied, and let $z_{\rho}$ be the solution to Problem \ref{PenDiscreteProbDef_MinMax}. There exists a neighbourhood $\mathcal{B}$ of $z_{\rho}$ such that, for any starting value $z^0\in\mathcal{B}$, the Newton sequence $(z^n)_{n\geq 0}$ generated by Algorithm \ref{NewtonAlg_Isaacs} is well defined, remains in $\mathcal{B}$ and converges to $z_{\rho}$\,; the rate of convergence is quadratic.
\end{theorem}
\begin{proof}
The result follows from Theorem 3 in \cite{Pang_NewtonsMethod} if $H$ is $B$-differentiable and has a strong and non-singular $F$-derivative at $z_{\rho}$\,. Recalling \eqref{SectionModNewton_Isaacs_Eq1}, we can argue as in the proof of Theorem \ref{NewtonAlg_HJB_PropTheorem} to obtain $B$-differentiability of $H$. From Theorem 2 in \cite{Pang_NewtonsMethod}, it follows that $H$ has a strong $F$-derivative, which is non-singular since $J_H(y)\in K^o_N$ for all $y\in\mathbb{R}^N$.
\end{proof}

We point out that, based on the assumptions of Theorem \ref{LineSearchNewton}, the Lipschitz property of $J_H$\,, additionally required in Theorem \ref{NewtonAlg_Isaacs_PropTheorem}, depends entirely on the regularity of $y\mapsto \max_{u\in\mathbf{U}}\{A_u\,y-b_u\}$ in \eqref{SectionModNewton_Isaacs_Eq1}.\bigskip

If we use $\Pi(\cdot)=\max\{\cdot,0\}$ (in which case, based on \eqref{SectionModNewton_Isaacs_Eq2}, $J_H$ is still well defined), \eqref{NewtonAlg_Isaacs_Eq1} simplifies, and we get the following special case of Algorithm \ref{NewtonAlg_Isaacs}.

\begin{alg}\label{ModifiedNewtonAlg_Isaacs} (Newton-like Method for the pen.\ HJB Obstacle Prob.)
For $i\in\mathcal{N}$ and $y\in\mathbb{R}^N$, define $u^{max}_i(y)\in\mathbf{U}$ as in \eqref{SectionModNewton_Isaacs_Eq0.9}, and set
\begin{equation*}
u^{max}(y):=(u^{max}_1(y)\,,u^{max}_2(y)\,,\ldots,u^{max}_N(y))^{tr}.
\end{equation*}
Furthermore, define $A^{+}(y)\in\mathbb{R}^{N\times N}$ and $b^{+}(y)\in\mathbb{R}^N$ to be matrix and vector consisting of
\begin{itemize}
\item rows $(\tilde{A})_i$ and $(\tilde{b})_i$\,, $i\in\mathcal{N}$, respectively, if $\max\{(\tilde{b}-\tilde{A}y)_i\,,0\}>0$
\item and having zero rows if $\max\{(\tilde{b}-\tilde{A}y)_i\,,0\}=0$.
\end{itemize} 
If $\Pi(y)=\max\{y,0\}$, we now have $J_H(y)=A_{u^{max}(y)} +  \rho A^{+}(y)\in K^o_N$\,, and \eqref{NewtonAlg_Isaacs_Eq1} becomes
\begin{multline*}
\big(A_{u^{max}(z^n)} +  \rho A^{+}(z^n)\big)(z^{n+1}-z^n)
=
-(A_{u^{max}(z^n)}(z^n)\,z^n - b_{u^{max}(z^n)})\\
-\rho A^{+}(z^n)z^n + \rho b^{+}(z^n),
\end{multline*}
which is equivalent to
\begin{align}
\big(A_{u^{max}(z^n)} +  \rho A^{+}(z^n)\big)z^{n+1}
=b_{u^{max}(z^n)} + \rho b^{+}(z^n).\label{ModifiedNewtonAlg_Isaac_Eq2}
\end{align}
\end{alg}

From \eqref{ModifiedNewtonAlg_Isaac_Eq2}, it is easy to see that Algorithm \ref{ModifiedNewtonAlg_Isaacs} is well defined.
The next theorem states that, given certain conditions on the matrix $\tilde{A}$, local
quadratic convergence of Algorithm \ref{ModifiedNewtonAlg_Isaacs} to the solution of Problem \ref{PenDiscreteProbDef_MinMax} for $\Pi(\cdot)=\max\{\cdot,0\}$ (see also Corollary \ref{ExistenceOfSolution_Cor_PiEqualsMaxFun}) can be guaranteed.

\begin{theorem}\label{NewtonAlg_Isaac_Properties}
Let $(z^n)^{\infty}_{n=0}$ be the sequence generated by Algorithm \ref{ModifiedNewtonAlg_Isaacs}.
There exists a constant $C>0$ such that, for every starting value $z^0$\,, it is $\|z^n\|_{\infty}\leq C$, $n\geq 1$.
Furthermore, if $\tilde{A}=I_N$\,, where $I_N\in\mathbb{R}^{N\times N}$ denotes the identity matrix, then there exists a neighbourhood
$\mathcal{B}$ of $z_{\rho}$ such that, for any starting value $z^0\in\mathcal{B}$, $(z^n)_{n\geq 0}$ remains in $\mathcal{B}$ and converges
to the solution $z_{\rho}$ at a quadratic rate.
\end{theorem}
\begin{proof}
The boundedness of $(z^n)_{n\geq 1}$ follows from \eqref{ModifiedNewtonAlg_Isaac_Eq2} by arguing as in the proof of Lemma \ref{NewtonAlg_FiniteTermination}. The quadratic local convergence property of $(z^n)_{n\geq 0}$ follows from Theorem 3.2 in \cite{Qi_NonsmoothNewton} if $H$ is semi-smooth and all $V\in\partial H(z_{\rho})$ are non-singular, where $\partial H$ denotes the generalised Jacobian as introduced in \cite{Clarke_OptimNonsmoothAnalysis}. By the assumptions at the beginning of this section, $H$ is continuously partially differentiable (and, in particular, semi-smooth) everywhere except on the set $\mathcal{D}:=\{y\in\mathbb{R}^N : \exists\ i\in\mathcal{N}\text{ s.t. }y_i=(\tilde{b})_i\}$, since the only critical term is $\max\{\tilde{b}-\tilde{A}y,0\}=\max\{\tilde{b}-y,0\}$.
Denoting the directional derivative at $y\in\mathbb{R}^N$ in direction $h\in\mathbb{R}^N$ by $H'(y;h)$, it can easily be verified that
\begin{equation}
\lim_{h\to 0}\frac{H'(y+h;h)-H'(y;h)}{\|h\|_{\infty}}=0,\quad y\in\mathcal{D},\label{NewtonAlg_Isaac_Properties_Eq1}
\end{equation}
which, by Theorem 2.3 in \cite{Qi_NonsmoothNewton}, gives semi-smoothness on $\mathcal{D}$, and, since the generalised Jacobian at $\tilde{b}$ is given by
\begin{equation*}
\partial H(\tilde{b})=\{A_{u^{max}(\tilde{b})} + \lambda \tilde{A} : \lambda\in[0,1]\}\subset K^o_N\,,
\end{equation*}
we see that all $V\in\partial H(\tilde{b})$ are in fact non-singular.
\end{proof}

%

\begin{remark}
Since, in practice, Problem \ref{PenDiscreteProbDef_MinMax} frequently results from a time-stepping routine, the solution
from the previous time step can be used as a starting value $z^0$, and, for small enough time-steps, $z^0\in\mathcal{B}$
as required by Theorems \ref{NewtonAlg_Isaacs_PropTheorem} and \ref{NewtonAlg_Isaac_Properties} can be expected to hold.
\end{remark}

\section{Evaluating the Continuous Control in Practice}\label{ContinuousControl_InPractice}

If we solve Problems \ref{DiscreteProbDef} and \ref{DiscreteProbDef_MinMax} using penalisation, we have to numerically deal with the continuous control, e.g. when computing $u^{min}_i(x^n)$ and $u^{max}_i(z^n)$ in \eqref{ModifiedNewtonAlg_Eq0.5} and \eqref{SectionModNewton_Isaacs_Eq0.9} for given $x^n$ and $z^n$, respectively.
Now, if the control is well behaved (as in the examples following later in this paper), the exact minimum/maximum can be found by differentiating and using analytical techniques; however, this is not necessarily always possible, i.e. we might have to approximate $(A_u)_{u\in\mathbf{U}}$ and $(b_u)_{u\in\mathbf{U}}$ by functions that are easier to handle numerically. The following remark states that such an approximation is legitimate within the framework of our algorithms.

\begin{remark} (Stability in the Control)\label{Remark_StabilityInTheControl}
Suppose we have families of functions $(A^{\epsilon}_u)_{u\in\mathbf{U}}$ and $(b^{\epsilon}_u)_{u\in\mathbf{U}}$\,, $\epsilon>0$, satisfying the following properties.
\begin{itemize}
\item For every $\epsilon>0$, $(A^{\epsilon}_u)_{u\in\mathbf{U}}$ and $(b^{\epsilon}_u)_{u\in\mathbf{U}}$ satisfy the same conditions as required in the definition of Problem \ref{DiscreteProbDef}
and in the setup of Sections \ref{SectionModNewton} and \ref{SectionModNewton_Isaacs}.
\item For every $u\in\mathbf{U}$, it is $\lim_{\epsilon\to 0}A^{\epsilon}_u = A_u$ and $\lim_{\epsilon\to 0}b^{\epsilon}_u = b_u$\,.
\end{itemize}
Suppose we solve Problems \ref{PenDiscreteProbDef} and \ref{PenDiscreteProbDef_MinMax} by using the approximations $(A^{\epsilon}_u)_{u\in\mathbf{U}}$ and $(b^{\epsilon}_u)_{u\in\mathbf{U}}$ for some $\epsilon>0$, running any of the algorithms discussed in Sections \ref{SectionModNewton} and \ref{SectionModNewton_Isaacs}, obtaining solutions $x^{*,\epsilon}$ and $z^{*,\epsilon}$, respectively.
In the limit $\epsilon\to 0$, we have $x^{*,\epsilon}\to x^*$ and $z^{*,\epsilon}\to z^*$, where $x^*$ and $z^*$, respectively, denote the solutions to Problems \ref{PenDiscreteProbDef} and \ref{PenDiscreteProbDef_MinMax} for $(A_u)_{u\in\mathbf{U}}$ and $(b_u)_{u\in\mathbf{U}}$.
\end{remark}
\begin{proof}
For fixed $\epsilon>0$, all results from the previous sections hold, and it only remains to show that $x^{*,\epsilon}\to x^*$ and $z^{*,\epsilon}\to z^*$ as $\epsilon\to 0$. Since all involved functions are continuous on the compact interval $\mathbf{U}$, the convergence
$\lim_{\epsilon\to 0}A^{\epsilon}_u = A_u$ and $\lim_{\epsilon\to 0}b^{\epsilon}_u = b_u$ is uniform in $u\in\mathbf{U}$; hence, we can reproduce results of Corollary \ref{BoundednessofMatrixandInv_Cor} with constants independent of $\epsilon>0$, which, following Lemmas \ref{PenDiscreteProbDef_BoundedIndepRho} and \ref{PenDiscreteProbDef_MinMax_BoundednessINdepOfRhoandPi}, means that $x^{*,\epsilon}$ and $z^{*,\epsilon}$ are bounded independently of $\epsilon>0$. We then infer the existence of converging subsequences, not notationally distinguished, of $(x^{*,\epsilon})_{\epsilon>0}$ and $(z^{*,\epsilon})_{\epsilon>0}$\,, respectively. Since the convergence in $\epsilon$ is uniform in $u\in\mathbf{U}$, we may swap ``$\lim_{\epsilon\to 0}$''
and ``$\max_{u\in\mathbf{U}}$'' in expressions \eqref{PenDiscreteProbDef_Eq1} and \eqref{PenDiscreteProbDef_MinMax_Eq1}, and we see that the limits of the sequences $(x^{*,\epsilon})_{\epsilon>0}$ and $(z^{*,\epsilon})_{\epsilon>0}$ do indeed solve Problems \ref{PenDiscreteProbDef} and \ref{PenDiscreteProbDef_MinMax}, respectively. The uniqueness of the solutions (see Lemmas \ref{PenDiscreteProbDef_Uniq} and \ref{PenDiscreteProbDef_MinMax_Uniq}) means that not only subsequences but the whole sequences converge.
\end{proof}

We point out that the rather general definition of approximating functions used in the above remark includes most common numerical approximations, e.g. piecewise constant, piecewise linear and other finite subspace approximations.

%
%

\section{Numerical Results}\label{Part_Numerics}

Finally, we come to test and analyse the applicability of the previously introduced numerical techniques by solving
two models from mathematical finance, presented in Sections \ref{Section_Examples_Indifference} and \ref{Section_Examples_EarlyExercise}, which lead directly to equations as given in Problems \ref{DifferentialOperatorProblem} and \ref{DifferentialOperatorProblem_MinMax}, respectively.

\subsection{Example: An Incomplete Market Investment Problem}\label{Section_Examples_Indifference}

In this section, we solve an incomplete market problem taken from \cite{Zariphopoulou_UnhedgeableRisk}.
More precisely, we look at an optimal investment model in which an agent has to distribute his money between a risk free
bond and a risky stock; the market incompleteness arises from the stochastic volatility of the stock price process.\bigskip

Let $b$, $a$, $\sigma:\mathbb{R}\to\mathbb{R}$ be bounded and globally Lipschitz functions, and suppose that there exists a constant $c$, independent of $y$, such that $\sigma(y)\geq c$.
Let $T>0$ be some finite time horizon.
Let $\mu>r>0$ and $B_0\,$, $S_0\,$, $Y_0> 0$.
Suppose we have a bond price process $(B_t)_{0\leq t\leq T}$\,, a stochastic volatility process $(Y_t)_{0\leq t\leq T}$ and a stock price process $(S_t)_{0\leq t\leq T}$ solving, respectively,
\begin{align*}
dB_t =&\ rB_t\,dt\,,\\
dY_t =&\ b(Y_t)\,dt + a(Y_t)\,dW^1_t\\
\text{and}\quad dS_t =&\ \mu S_t\,dt + \sigma(Y_t)S_t\,dW^2_t\,,
\end{align*}
where $(W^i_t)_{0\leq t\leq T}$\,, $i\in\{1,2\}$, are two Brownian motions defined on a probability space $(\Omega,\mathcal{F},P)$ with a correlation coefficient $\varrho\in[-1,1]$.\bigskip

Staying exactly in the framework of \cite{Zariphopoulou_UnhedgeableRisk}, we consider an investor who can invest in the stock and in the bond.
We suppose that the investor has initial wealth $X_0\geq0$ and that he may rebalance his portfolio $X_t=\pi^0_t+\pi_t$ at any time $t\in[0,T]$; here, $\pi^0_t$ and $\pi_t$ denote the amounts invested, respectively, in the bond and in the stock.
The investor's wealth process solves
\begin{equation*}
dX_t = rX_t\,dt + (\mu-r)\pi_t
\,dt + \sigma(Y_t)\pi_t\,dW^1_t\,,\quad t\in[0,T],
\end{equation*}
and must satisfy $X_t\geq 0$ for every $t\in[0,T]$.
We take the investor's utility function to be of CRRA-type and given by
\begin{equation*}
U(x)=\frac{1}{\gamma}x^{\gamma},\quad x\in\mathbb{R},
\end{equation*}
for some constant $0<\gamma<1$. Now, trying to maximise the final utility, the investor's value function is given by
\begin{equation}
\phi(x,y,t): = \sup_{\pi\in\mathcal{A}}\mathbb{E}[U(X_T)|X_t=x,Y_t=y],\quad (x,y,t)\in[0,\infty)\times\mathbb{R}\times[0,T],\label{Example_InvestorValueFunction}
\end{equation}
where $\mathcal{A}$ is the set of admissible trading strategies (for details, see \cite{Zariphopoulou_UnhedgeableRisk}). We cite the following result, which shows how this utility maximisation problem can be solved.

\begin{prop}\label{IncompleteMarketProb_EqFormulation}
The value function introduced in \eqref{Example_InvestorValueFunction} can be written as
\begin{equation*}
\phi(x,y,t) = \frac{x^{\gamma}}{\gamma} \varphi(y,t),\quad(x,y,t)\in[0,\infty)\times\mathbb{R}\times[0,T],
\end{equation*}
where $\varphi$ solves
\begin{multline}
\frac{1}{\gamma}\big[\varphi_t + \frac{1}{2}a^2(y)\varphi_{yy} + b(y)\varphi_y\big] + r\varphi \\ + \max_{u\in\mathbf{U}}\big[\frac{1}{2}(\gamma-1)\sigma^2(y)u^2\varphi+\varrho\sigma(y)a(y)u\varphi_y+(\mu-r)u\varphi\big]=0,\label{IncompleteMarketProb_EqFormulation_HJBEq}
\end{multline}
with $\varphi(y,T)\equiv 1$ and $\mathbf{U}\subset\mathbb{R}$ an appropriately chosen compact set, and $\tilde{\varphi}(\cdot,\cdot):= \varphi(\cdot,\cdot)^{\frac{1-\gamma+\varrho^2\gamma}{1-\gamma}}$ satisfies
\begin{multline}
\tilde{\varphi}_t + \frac{1}{2}a(y)^2\tilde{\varphi}_{yy} + \Big[b(y) + \varrho\frac{\gamma(\mu-r)a(y)}{(1-\gamma)\sigma(y)}\Big]\tilde{\varphi}_y\\
+ \frac{\gamma(1-\gamma+\varrho^2\gamma)}{1-\gamma}\Big[r + \frac{(\mu-r)^2}{2\sigma^2(a)(1-\gamma)}\Big]\tilde{\varphi}=0.\label{IncompleteMarketProb_EqFormulation_LinParPDE}
\end{multline}
The notion of solution used in this context has to be understood in the viscosity sense (cf.\,\cite{UsersGuide_ViscositySols}), and $\varphi$, $\tilde{\varphi}$ are the unique solutions to \eqref{IncompleteMarketProb_EqFormulation_HJBEq} and \eqref{IncompleteMarketProb_EqFormulation_LinParPDE}, respectively.
\end{prop}
\begin{proof}
The main result can be found in \cite{Zariphopoulou_UnhedgeableRisk}. See \cite{Barles_ErrorBoundsMonotoneApproxSchemes} for details
on the existence of the viscosity solutions.
\end{proof}


In \eqref{IncompleteMarketProb_EqFormulation_HJBEq}, we make the assumption of $\mathbf{U}$ being a compact set since, for every time $t\in[0,T)$, the maximum in \eqref{IncompleteMarketProb_EqFormulation_HJBEq} is assumed at $u=\pi^*_t$\,, where $(\pi^*_t)_{0\leq t\leq T}$ is the investor's optimal trading policy (cf.\,\cite{Zariphopoulou_UnhedgeableRisk}), which should not reach infinity in a meaningful financial model.\bigskip

Clearly, from Theorem \ref{IncompleteMarketProb_EqFormulation}, to find $\phi$, we need to compute $\varphi$ or $\tilde{\varphi}$ and solve either equation \eqref{IncompleteMarketProb_EqFormulation_LinParPDE} or \eqref{IncompleteMarketProb_EqFormulation_HJBEq}; we have deliberately chosen a problem which can be linearised such that we can obtain a reference solution by standard methods. In the next few sections, we will select parameters and present and compare several approaches of computing $\phi$.

\subsubsection{Choosing Model Parameters and Functions}
We set $r=0.3$, $\mu=0.7$, $\varrho = -0.2$, $\gamma = 0.5$, and $T=1$. Furthermore, we introduce $y_{min} :=\kappa := 0.1$ and $y_{max}:=1$, and, for $y\in[y_{min}\,,y_{max}]$, we use
\begin{align*}
a(y) =& -2.5(y-0.5-0.5\kappa)^2 + 2.5(-0.5+0.5\kappa)^2,\\
b(y) =& -y+0.55,\\
\text{and}\quad\sigma(y) =&\ y.
\end{align*}
Functions $a$, $b$ and $\sigma$ are shown in Figure \ref{fig:Fun_a_b_sigma}; they satisfy the technical conditions listed in the previous section and guarantee $(Y_t)_{t\geq 0}\subset [\kappa,1]$.

\begin{figure}[ht]
\centering
\includegraphics[width=8cm,height=8cm]{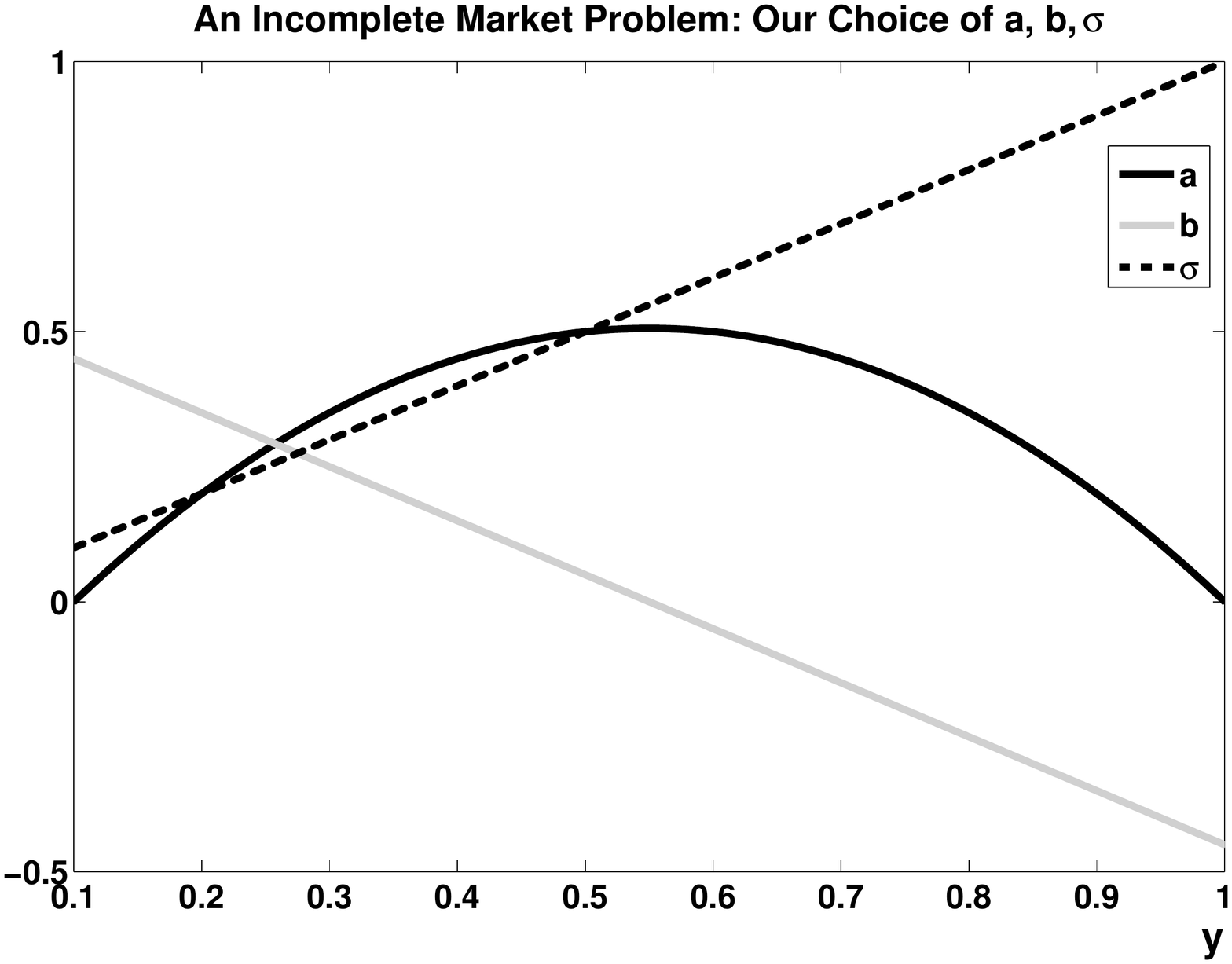}
\caption{The functions $a$, $b$ and $\sigma$ as chosen for the volatility process in the investment problem introduced in Section \ref{Section_Examples_Indifference}. As can easily be observed, whenever $y$ approaches $\kappa$ (or $1$), $a(y)$ goes to zero and $b(y)$ is positive (negative); this guarantees that the diffusion $dY_t = b(Y_t)\,dt + a(Y_t)\,dW^1_t$ stays in $[\kappa,1]$. Additionally, when (numerically) solving the two equations \eqref{IncompleteMarketProb_EqFormulation_LinParPDE} and \eqref{IncompleteMarketProb_EqFormulation_HJBEq} on the interval $[\kappa,1]$, outgoing characteristics (cf.\,\cite{Strikwerda_FinDiff_PDE}) replace our boundary conditions at $\kappa$ and $1$.}
\label{fig:Fun_a_b_sigma}
\end{figure}

\subsubsection{Discretisation of the Continuous Equations}\label{DiscretisationsOfEquations_ContControl}
Numerically, we solve equations \eqref{IncompleteMarketProb_EqFormulation_LinParPDE} and \eqref{IncompleteMarketProb_EqFormulation_HJBEq} backwards in time, starting at expiry $T$, on a grid
$\{ (ih+\kappa,jk) : 0\leq i\leq N,\ 0\leq j\leq M\}$, where $h:= (1-\kappa)/N$ and $k:=T/M$. We perform a fully implicit finite difference discretisation, using one-sided differences for all first derivatives (including the time derivative) and central differences for all second derivatives. (For a general overview of basic finite difference concepts for PDEs, e.g.\ see \cite{Seydel_ToolsCompFinance}.) In particular, when discretising $\mathfrak{d}\varphi_y$ (or $\mathfrak{d}\tilde{\varphi}_y$), where $\mathfrak{d}=\mathfrak{d}(y,u)$ denotes the combined coefficient of the first $y$-derivative in \eqref{IncompleteMarketProb_EqFormulation_LinParPDE} (or \eqref{IncompleteMarketProb_EqFormulation_HJBEq}), we switch between left-sided and right-sided differences in accordance with a positive or negative sign of the coefficient $\mathfrak{d}$; this way, we can guarantee the following two properties.
\begin{itemize}
\item The tridiagonal discretisation matrix implied by our fully implicit scheme has positive entries on the diagonal and non-positive entries on the upper/lower diagonals.
\item Since $a(\kappa)=a(1)=0$ and $b(\kappa)>0>b(1)$, the boundary conditions at $\kappa$ and $1$ are replaced by finite differences pointing inwards. (cf.\,\cite{Strikwerda_FinDiff_PDE}).
\end{itemize}
Proceeding as just described, the linear parabolic PDE in \eqref{IncompleteMarketProb_EqFormulation_LinParPDE} is approximated by a simple linear system of equations. Furthermore, taking $\mathbf{U}$ to be $[-l,l]$, with $l=150$, our discretisation of equation \eqref{IncompleteMarketProb_EqFormulation_HJBEq} matches Problem \ref{DiscreteProbDef}, with \eqref{DiscreteProbDef_ConFunDef2} and \eqref{DiscreteProbDef_ConFunDef1} satisfying all assumptions.

\begin{remark}
Based on our choice of functions $a$ and $b$ (cf.\,Figure \ref{fig:Fun_a_b_sigma}), the diffusion $(Y_t)_{0\leq t\leq T}$
will always stay in $[\kappa,1]$. Conceptually, this means that no `outside' information -- like Dirichlet boundary conditions -- is required for the PDE to be completely specified on the interval $[\kappa,1]$, since the flow of information can be thought of as coming from `within'. Mathematically, this means that, by taking $a(y):=0$, $y\in\mathbb{R}\backslash [\kappa,1]$, and using one-sided inwards pointing finite difference stencils at $\kappa$ and 1, we obtain a consistent (and monotone) discretisation scheme without requiring any Dirichlet boundary conditions.
\end{remark}

\begin{remark}\label{ConvergenceRemark_HJBExample}
It is generally non-trivial to prove convergence of finite difference schemes applied to a (possibly nonlinear) PDE for which only viscosity solutions can be shown to exist. The standard reference is \cite{BarlesMainArticle}, where -- loosely speaking -- it is shown that every stable and monotone discretisation converges if the equation satisfies a strong comparison principle (cf.\,\cite{UsersGuide_ViscositySols}); other (more specific) approaches include \cite{Krylov_RateOfConcergence_VariableCoeff,Barles_OntheConvergenceRate_ApproxHJBEq,Jakobsen_OnRateConvergence_ApproxSchemes_BellmanEq,Krylov_RateOfConvergence_LipschitzCoeff,Barles_ErrorBoundsMonotoneApproxSchemes,Jakobsen_ErrorEstimates_BellmanEquations_ControlledJumpDiffusion}.
For the current example, consistency is straightforward to prove (e.g.\,cf.\,\cite{WitteReisinger_PenaltyScheme_DiscreteControlledHJBEquations}), stability can be shown following \cite{Forsyth_Controlled_HJB_PDEs_Finance} since we have what they call a ``positive coefficient discretisation'', and a strong
comparison result can be found in \cite{Barles_ErrorBoundsMonotoneApproxSchemes}; hence, altogether, the results
of \cite{BarlesMainArticle} are applicable and convergence to the unique viscosity solutions
of \eqref{IncompleteMarketProb_EqFormulation_LinParPDE} and \eqref{IncompleteMarketProb_EqFormulation_HJBEq} can be guaranteed.
\end{remark}

\subsubsection{Solution of the Discrete Systems}\label{SolvingDiscrSystems_ContControl}

In this section, we will compare the following three ways of numerically solving the incomplete market problem represented by the equations in Theorem \ref{IncompleteMarketProb_EqFormulation}. All computations are done in Matlab.
\begin{itemize}
\item We solve the linear system of equations resulting -- for every time step -- from a discretisation of the linear parabolic PDE in \eqref{IncompleteMarketProb_EqFormulation_LinParPDE}.
\item We solve Problem \ref{DiscreteProbDef} (corresponding to one time step of a fully implicit discretisation of the non-linear parabolic PDE in \eqref{IncompleteMarketProb_EqFormulation_HJBEq}) by
\begin{itemize}
\item the penalty method devised in Sections \ref{Section_PenalisationDiscrProb} and \ref{SectionModNewton} of this paper, and by
\item the method of policy iteration.
\end{itemize}
\end{itemize}
The use of policy iteration has been studied in \cite{Forsyth_Controlled_HJB_PDEs_Finance,Bokanowski_Howard_Algorithm} and is briefly summarised in Appendix \ref{Appendix_PolicyIteration_ContinuouslyControlled}.
Following Section \ref{ContinuousControl_InPractice}, we
approximate $\mathbf{U}=[-l,l]=[-150,150]$ by $\widetilde{\mathbf{U}}:=\{-150 + \frac{3r}{10} : r = 0,1,\ldots,1000\}$, thereby discretising
$\mathbf{U}$ by using a very fine grid of 1001 points. (Effectively, in the notation of Remark \ref{Remark_StabilityInTheControl}, we approximate $(A_u)_{u\in\mathbf{U}}$ and $(b_u)_{u\in\mathbf{U}}$ by piecewise constant functions.)
Numerically, when computing candidate solutions to
equations \eqref{DiscreteProbDef_Eq1} and \eqref{PenDiscreteProbDef_Eq1} by policy iteration and Algorithm \ref{ModifiedNewtonAlg}, respectively, we terminate the iterations according to the following two checks of accuracy, using $tol=$1e-08.
\begin{itemize}
\item For \eqref{DiscreteProbDef_Eq1}, we terminate if our candidate solution $x^{n}\in\mathbb{R}^N$ satisfies
\begin{equation}
\frac{\|A_{\tilde{u}^{sup}(x^n)}\,x^n-b_{\tilde{u}^{sup}(x^n)}\|_{\infty}}{\|b_{\tilde{u}^{sup}(x^n)}\|_{\infty}}\leq tol,\label{SolvingDiscrSystems_ContControl_Eq0.9}
\end{equation}
where $\tilde{u}^{sup}(x^n)$ satisfies $A_{\tilde{u}^{sup}(x^n)}\,x^n-b_{\tilde{u}^{sup}(x^n)} = \min\{A_{\tilde{u}}\,x^n-b_{\tilde{u}} : \tilde{u}\in\widetilde{\mathbf{U}}\}$.
\item For \eqref{PenDiscreteProbDef_Eq1}, we terminate if our candidate solution $x^{n}_{\rho}\in\mathbb{R}^N$ satisfies
\begin{equation}
\frac{\|(A_{u_0}\,x^n_{\rho} - b_{u_0})-\rho\max\{b_{\tilde{u}^{sup}_\rho(x^n_{\rho})}-A_{\tilde{u}^{sup}_\rho(x^n_{\rho})}\,x^n_{\rho}\,,0\}\|_{\infty}}{\|b_{u_0} + \rho\,b^+_{\tilde{u}^{sup}_{\rho}(x^n_{\rho})}\|_{\infty}}\leq tol,\label{SolvingDiscrSystems_ContControl_Eq1}
\end{equation}
where $\tilde{u}^{sup}_{\rho}(x^n)$ satisfies 
$b_{\tilde{u}^{sup}_\rho(x^n)}-A_{\tilde{u}^{sup}_\rho(x^n)}\,x^n_{\rho} = \max\{b_{\tilde{u}_\rho}-A_{\tilde{u}_\rho}\,x^n_{\rho} : \tilde{u}_\rho\in\widetilde{\mathbf{U}}\}$ and
\begin{equation*}
b^+_{\tilde{u}^{sup}_{\rho}(x^n_{\rho})}:= 
\begin{cases} b_{\tilde{u}^{sup}_{\rho}(x^n_{\rho})} & \text{if $b_{\tilde{u}^{sup}_\rho(x^n)}-A_{\tilde{u}^{sup}_\rho(x^n)}\,x^n_{\rho}>0$,}
\\
0 &\text{else.}
\end{cases}
\end{equation*}
\end{itemize}\bigskip


Figure \ref{fig:u0_Dependence} shows the impact of the a priori unspecified parameter $u_0$ in Problem \ref{PenDiscreteProbDef} on the penalisation error. We first remark that, if $u_0$ could be chosen to be the optimal control of the discretised problem, the penalisation error in Problem \ref{PenDiscreteProbDef} would be identical to zero by construction. Now, clearly, the optimal control is unknown, and generally a function of the state variable and time, and, therefore, a constant value $u_0$ generally gives a non-zero (in fact, negative) penalisation error, which we know to converge to zero of first order in $1/\rho$. This is the underlying convergence mechanism of the proposed method and does not require a cunning choice of $u_0$. Figure \ref{fig:u0_Dependence} does show, however, that the penalisation error for fixed $\rho$ can be reduced by diligent choice of $u_0$. It is thereby sufficient to choose $u_0$ of the same order of magnitude as the optimal control. One can thus take advantage of a priori knowledge of the approximate control size. If such an estimate is not available, one might first determine a rough approximation by producing a crude version of Figure \ref{fig:u0_Dependence} on a coarse mesh -- coarse in parameter space as well as time and state space -- which is computationally cheap, and pick $u_0$ accordingly. We do not take advantage of this information in the following computations, and, throughout, we use $u_0 = -l =-150$, which appears to be the worst-case choice.
We also tested the impact of $u_0$ on the Newton method, and found the required number of iterations virtually unaffected in all settings.\bigskip


We use the numerical solution of the linear parabolic PDE \eqref{IncompleteMarketProb_EqFormulation_LinParPDE} as a \textit{reference solution} for the incomplete market problem. For $M=N=200$ and $\rho = 1e06$, measured in the maximum norm, the difference between the numerical solution of \eqref{IncompleteMarketProb_EqFormulation_LinParPDE} and the penalty approximation \eqref{PenDiscreteProbDef_Eq1} is 2e-03, and the difference between the penalty approximation \eqref{PenDiscreteProbDef_Eq1}
and the policy iteration solution of \eqref{DiscreteProbDef_Eq1} is 2e-04.\bigskip

The results of our numerical tests are summarised in Figures \ref{fig:HJB_ExactSol_and_PenSol} and \ref{fig:HJB_PenaltyConvergence} and in Tables \ref{tab:HJB_NumofIterations} and \ref{tab:HJB_CompTimes}.\bigskip

In Figure \ref{fig:HJB_ExactSol_and_PenSol}, we see the penalty approximation for $\rho=1e03$. Given that the penalty parameter is still relatively small, the penalty approximation is still below the reference solution, but, clearly, the curves are similarly shaped already.
In Figure \ref{fig:HJB_PenaltyConvergence}, we measure rate of the convergence in $\rho$; more precisely, for different sizes of $\rho$, starting out with the $t=k$ value of the reference solution (i.e. the value at the penultimate time step), we compute a time zero value using the penalty scheme (thus solving exactly one discrete LCP) and compare it to the time zero value of the reference solution. As expected, based on Theorem \ref{PenaltyConvergenceToTrueSol_ErrorEstimate}, the convergence in $\rho$ is of first order.
In Table \ref{tab:HJB_NumofIterations}, we see the number of iterations needed by Algorithm \ref{ModifiedNewtonAlg} and policy
iteration, averaged over all time steps; for the two schemes, the numbers are almost identical, and -- in both cases -- we never need more than two iterations to reach the desired accuracy $tol$. Finally, in Table \ref{tab:HJB_CompTimes}, we see the computation times for
the two schemes, which, as is to be expected based on the iteration numbers, are virtually the same.

\begin{figure}[t]
\centering
\includegraphics[width=8cm,height=8cm]{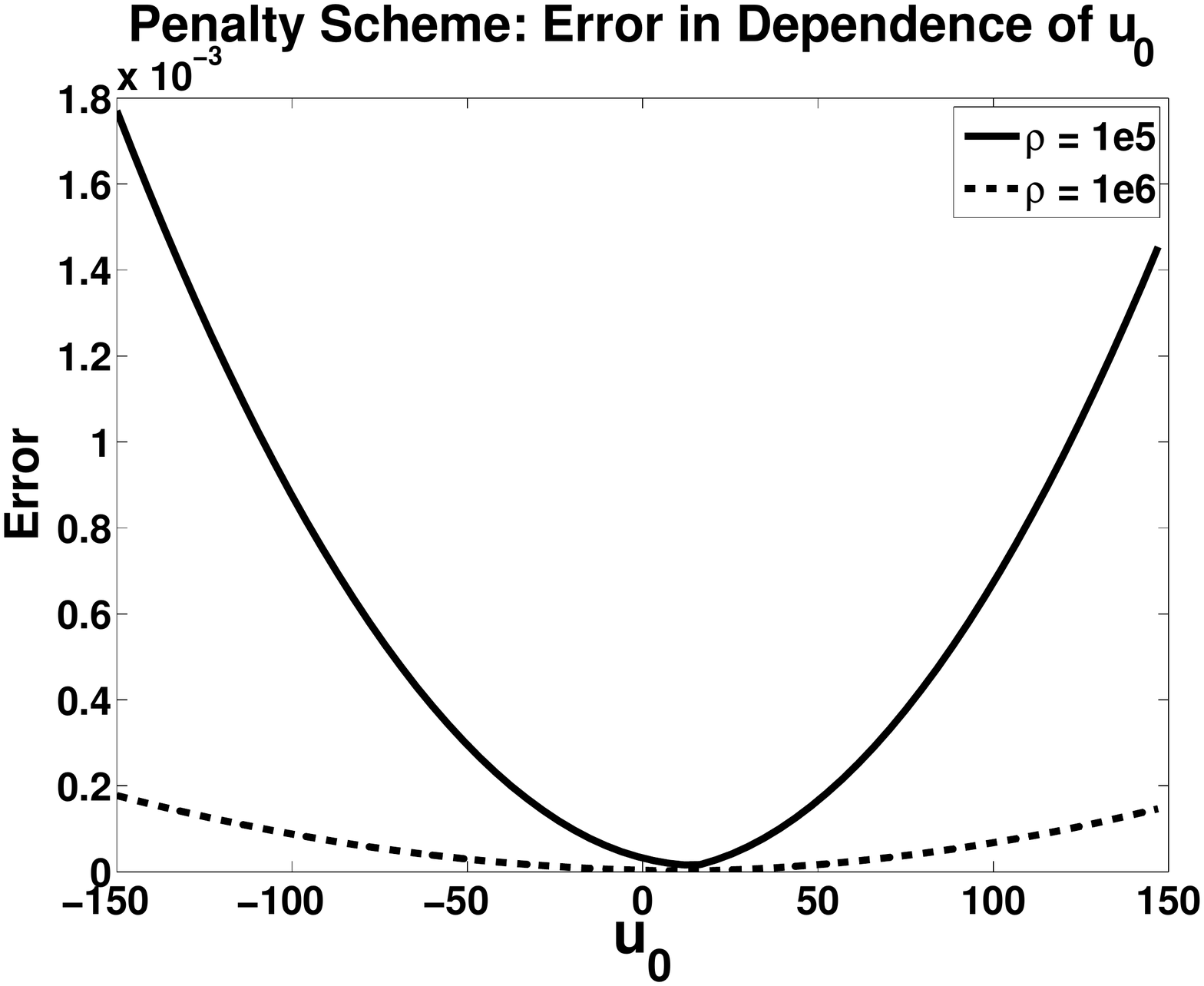}
\caption{The incomplete market investment problem of Section \ref{Section_Examples_Indifference}.
For $M=N=200$ and different choices of $u_0$\,, we see the difference between the solution computed by the penalty
scheme and the policy iteration method. The error is measured in the max-norm. We observe that, for the quality of the penalty approximation, it is advantageous to pick $u_0$ close to the optimal control value.}
\label{fig:u0_Dependence}
\end{figure}

\begin{figure}[ht]
\centering
\includegraphics[width=8cm,height=8cm]{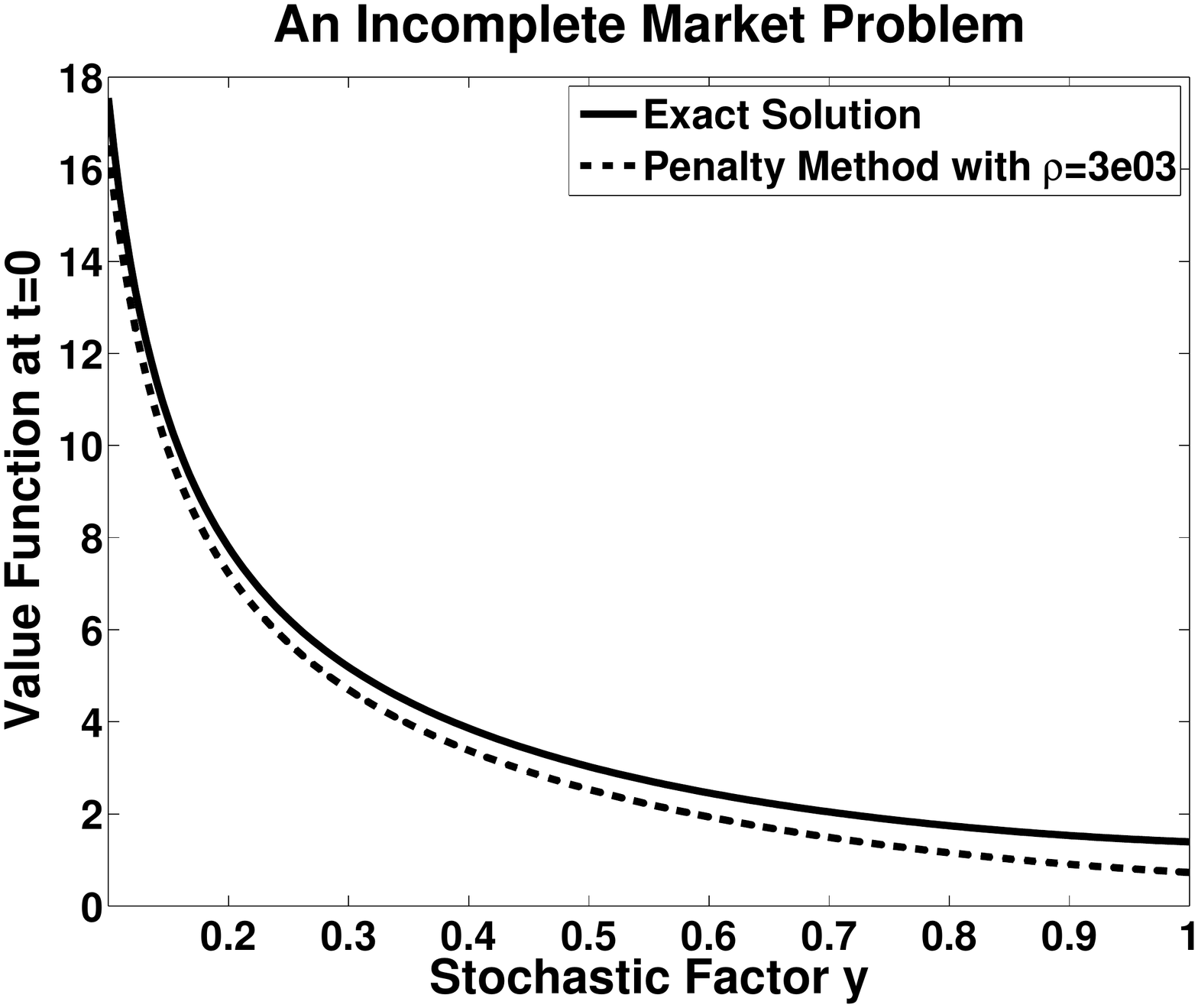}
\caption{The incomplete market investment problem of Section \ref{Section_Examples_Indifference}. For $M=N=200$, we see the solution to \eqref{IncompleteMarketProb_EqFormulation_LinParPDE}, referred to as ``reference solution'' since it avoids the treatment of non-linearities, and the solution to the penalised equation \eqref{PenDiscreteProbDef_Eq1} for $\rho=1e03$; for the current choice of $\rho$, the penalty approximation is still coarse.}
\label{fig:HJB_ExactSol_and_PenSol}
\end{figure}

\begin{figure}[ht]
\centering
\includegraphics[width=8cm,height=8cm]{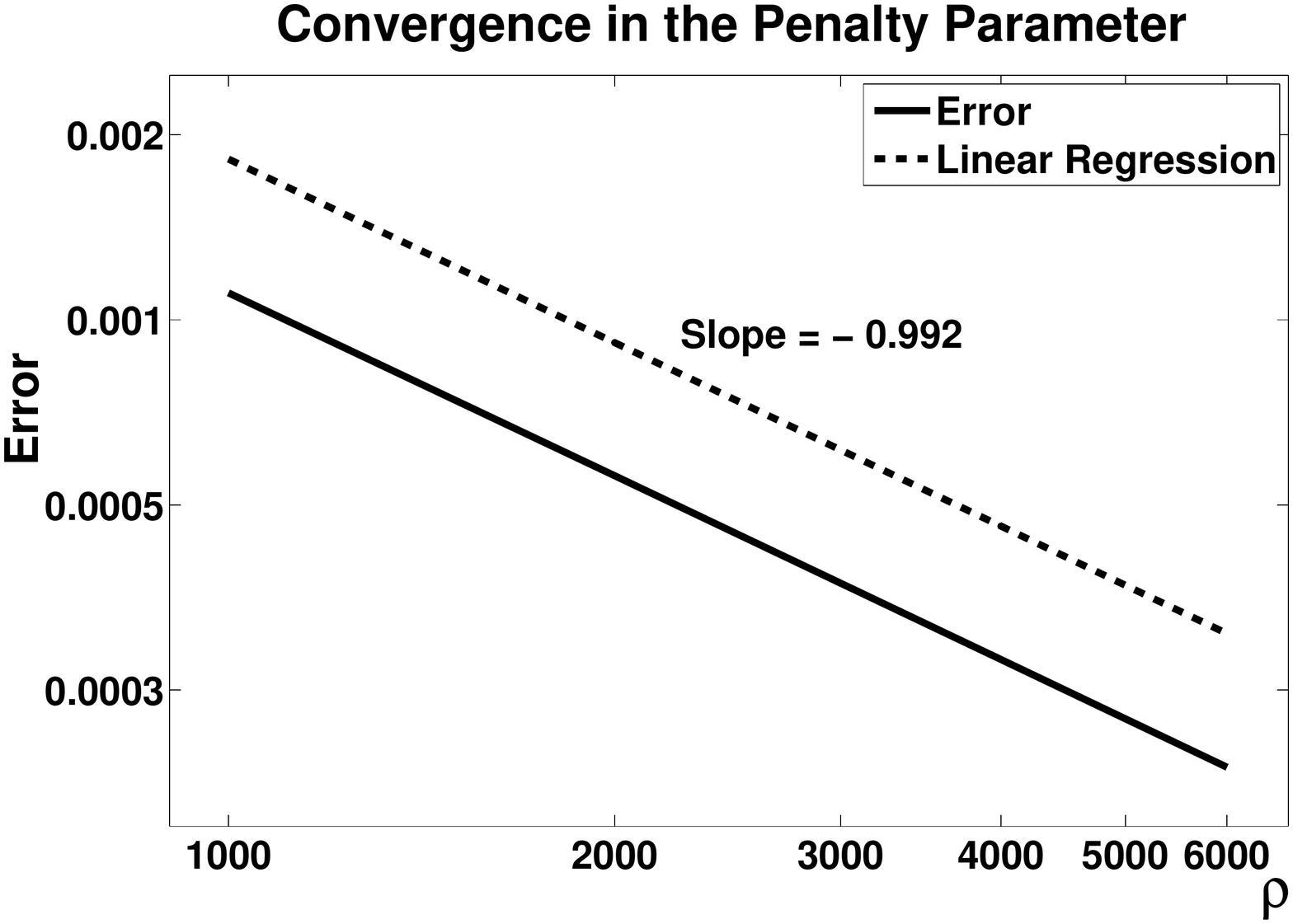}
\caption{The incomplete market investment problem of Section \ref{Section_Examples_Indifference}. For M=N=200, we measure the speed of convergence in the penalty parameter $\rho$. The error is measured in the $max$-norm. The plot is log-log, and we observe
a convergence rate of 0.992, close to one, confirming the results of Theorem \ref{PenaltyConvergenceToTrueSol_ErrorEstimate}.}
\label{fig:HJB_PenaltyConvergence}
\end{figure}

\begin{table}[t]
\begin{tabular}{|l|c|c|c|c|c|}
\hline
\textit{Policy Iteration} & $n=1$ & $n=2$\\
  \hline
$M$, $N=50$ & 6$\%$ & 94$\%$\\
  \hline
$M$, $N=200$ & 11$\%$ & 89$\%$\\
\hline
$M=200$, $N=50$ & 53$\%$ & 47$\%$\\
\hline
$M=50$, $N=200$ & - & 100$\%$\\
\hline
\hline
\textit{Penalty Method} $(\rho = 4e03)$ & $n=1$ & $n=2$\\
  \hline
$M$, $N=50$ & 8$\%$ & 92$\%$ \\
  \hline
$M$, $N=200$ & 13$\%$ & 87$\%$ \\
\hline
$M=200$, $N=50$ & 49.5$\%$ & 50.5$\%$ \\
\hline$M=50$, $N=200$ & - & 100$\%$ \\
\hline
\hline
\textit{Penalty Method} $(\rho = 1e06)$ & $n=1$ & $n=2$\\
  \hline
$M$, $N=50$ & 6$\%$ & 94$\%$ \\
  \hline
$M$, $N=200$ & 11$\%$ & 89$\%$ \\
  \hline
$M=200$, $N=50$ & 55$\%$ & 45$\%$ \\
  \hline
$M=50$, $N=200$ & - & 100$\%$ \\
  \hline
\end{tabular}
\caption{The incomplete market investment problem of Section \ref{Section_Examples_Indifference}. For different time and space grids, we see the number of iterations needed
by penalty approximation -- or, more precisely, by Algorithm \ref{ModifiedNewtonAlg} when solving \eqref{PenDiscreteProbDef_Eq1} -- and policy iteration. For the different schemes, the numbers are very similar, and, generally, we never need more than two steps.}
\label{tab:HJB_NumofIterations}
\end{table}

\begin{table}[t]
\begin{tabular}{|c|c|c|c|}
\hline
\textit{Grid Size} & \textit{Policy} & \textit{Penalty} $(\rho = 4e03)$ & \textit{Penalty} $(\rho = 1e06)$ \\
  \hline
$M$, $N=50$ & 1.08$s$ & 1.08$s$ & 1.09$s$ \\
  \hline
$M$, $N=200$ & 35.07$s$ & 34.76$s$ & 35.16$s$ \\
\hline
$M=200$, $N=50$ & 2.74$s$ & 2.82$s$ & 2.76$s$ \\
\hline
$M=50$, $N=200$ & 10.89$s$ & 10.88$s$ & 10.91$s$ \\
\hline
\end{tabular}
\caption{The incomplete market investment problem of Section \ref{Section_Examples_Indifference}. For different time and space grids, the computation times needed
by penalty approximation and policy iteration. In all cases, the computational effort of the two schemes is very similar.}
\label{tab:HJB_CompTimes}
\end{table}

\subsection{Example: Early Exercise Options in an Incomplete Market}\label{Section_Examples_EarlyExercise}

In this section, we value an early exercise contract in an incomplete market, in which the incompleteness stems from the fact that
the asset on which the early exercise contracts are written is not traded; the example in its entirety is taken from
\cite{Oberman_Thalia_EarlyExerciseInbcompleteMarkets}, and we use it to demonstrate that the arising non-linear equations
can be solved by a fully implicit finite difference scheme when using the results of Section \ref{Penalise_Discr_Obst_prob}.\bigskip

Let $b$, $a:\mathbb{R}\to\mathbb{R}$ be bounded and globally Lipschitz functions.
Let $T>0$ be some finite time horizon.
Let $\mu$, $\sigma>0$ and $S_0>0$, $Y_0\in\mathbb{R}$.
Suppose we have a traded asset price process $(S_t)_{0\leq t\leq T}$ and a non-traded asset price process $(Y_t)_{0\leq t\leq T}$ solving, respectively,
\begin{align*}
dS_t =&\ \mu S_t\,dt + \sigma S_t\,dW^1_t\\
\text{and}\quad dY_t =&\ b(Y_t)\,dt + a(Y_t)\,dW^2_t\,,
\end{align*}
where $(W^i_t)_{0\leq t\leq T}$\,, $i\in\{1,2\}$, are two Brownian motions defined on a probability space $(\Omega,\mathcal{F},P)$ with a correlation coefficient $\varrho\in[-1,1]$. Additionally, we assume the existence of a riskless bond with interest rate $r=0$.\bigskip

Similar to the previous example, we consider an investor who can invest in the traded asset and in the bond.
We suppose that the investor has initial wealth $X_0\in\mathbb{R}$ and that he may rebalance his portfolio $X_t=\pi^0_t+\pi_t$ at any time $t\in[0,T]$; here, $\pi^0_t$ and $\pi_t$ denote the amounts invested, respectively, in the bond and in the stock.
The investor's wealth process solves
\begin{equation*}
dX_t = \mu\pi_t\,dt + \sigma\pi_t\,dW^1_t\,,\quad t\in[0,T].
\end{equation*}
We take the investor's utility function to be of exponential type and given by
\begin{equation*}
U(x)= -e^{-\gamma x},\quad x\in\mathbb{R},
\end{equation*}
with risk aversion parameter $\gamma>0$. Now, we suppose the investor holds an early exercise contract with payoff $P(y)$, $y\in\mathbb{R}$, on the non-traded asset, and we would like to find the indifference price (cf.\,\cite{Oberman_Thalia_EarlyExerciseInbcompleteMarkets}) of the instrument; we cite the following result.

\begin{prop}\label{IncompleteMarketProb_EarlyExercise_EqFormulation}
The buyer's early exercise indifference price $\psi(y,t)$, where $(y,t)\in\mathbb{R}\times [0,T]$, is the unique bounded viscosity solution to
\begin{equation}
\min\big\{-\psi_t -\mathcal{L}^b\psi + \frac{1}{2}\gamma(1-\varrho^2)a^2(y)\psi^2_y\,, \psi-P(y) \big\} = 0,\label{IncompleteMarketProb_EarlyExercise_EqFormulation_Eq1}
\end{equation}
where $\psi(\cdot,T)=P(\cdot)$ and
\begin{equation*}
\mathcal{L}^b\psi := \frac{1}{2}a^2(y)\psi_{yy} + \big(b(y)-\varrho\frac{\mu}{\sigma}a(y)\big)\psi_y\,.
\end{equation*}
\end{prop}
\begin{proof}
The main result can be found in \cite{Oberman_Thalia_EarlyExerciseInbcompleteMarkets}. See \cite{Barles_ErrorBoundsMonotoneApproxSchemes} for details
on the existence of the viscosity solution.
\end{proof}

It can easily be shown that $\psi^2_y = \max_{u\in\mathbb{R}}\{2u\psi_y-u^2\}$, and, if we assume $\psi_y$ to be bounded, 
\eqref{IncompleteMarketProb_EarlyExercise_EqFormulation_Eq1} can be rewritten as
\begin{equation}
\min\big\{ \max\{\mathcal{L}^b_u\psi : u\in\mathbf{U}\}, \psi-P(y) \big\} = 0,\label{IncompleteMarketProb_EarlyExercise_EqFormulation_Eq2}
\end{equation}
where $\mathbf{U}\subset\mathbb{R}$ is a suitably chosen compact set and, for $u\in\mathbf{U}$, we define
\begin{equation*}
\mathcal{L}_u^b\psi := -\psi_t - \frac{1}{2}a^2(y)\psi_{yy} - \big(b(y)-\varrho\frac{\mu}{\sigma}a(y)\big)\psi_y
+ \frac{1}{2}\gamma(1-\varrho^2)a^2(y)(2u\psi_y-u^2).
\end{equation*}

Hence, to compute the early exercise indifference price of the considered option, we have to solve \eqref{IncompleteMarketProb_EarlyExercise_EqFormulation_Eq2}, which has the same structure as \eqref{DifferentialOperatorProblem_MinMax_Eq1}.

\subsubsection{Choosing Model Parameters and Functions}

We set $\mu/\sigma = 1$, $\varrho = 0.1$, $\gamma = 1$ and $T=1$. Furthermore, we introduce $y_{min} := 0$ and $y_{max}:=5$, and, for $y\in[y_{min}\,,y_{max}]$, we use
\begin{align*}
a(y) =&\ y,\\
b(y) =&\ 0.3y\\
\text{and}\quad P(y) =&\ \max\{1-y,0\}.
\end{align*}
In particular, the choice of $P(\cdot)$ means that we are dealing with an American put with strike one.

\subsubsection{Discretisation of the Continuous Equations and Solution of the Discrete Systems}\label{DiscretisationsOfEquations_EarlyExercise}

Numerically, we solve \eqref{IncompleteMarketProb_EarlyExercise_EqFormulation_Eq1} and \eqref{IncompleteMarketProb_EarlyExercise_EqFormulation_Eq2} similarly to Section \ref{DiscretisationsOfEquations_ContControl}, i.e. we proceed backwards in time, starting at expiry $T$, on a grid
$\{ (ih,jk) : 0\leq i\leq N,\ 0\leq j\leq M\}$, where $h:= 5/N$ and $k:=T/M$. In space, we again apply a finite difference discretisation guaranteeing for the discretisation matrices to be in $K^o_N$\,. Since we are dealing with a put option, we use
$\psi(0)=1$ and $\psi(5)=0$ as boundary conditions, and we take the
set $\mathbf{U}$ in \eqref{IncompleteMarketProb_EarlyExercise_EqFormulation_Eq2} to be $[-1,0]$.
We compare the following three numerical approaches in Matlab.
\begin{itemize}
\item For \eqref{IncompleteMarketProb_EarlyExercise_EqFormulation_Eq1}, we use an explicit time stepping scheme, meaning all non-linearities can be dealt with easily.
\item We solve Problem \ref{DiscreteProbDef_MinMax} (corresponding to one time step of a fully implicit discretisation of the non-linear parabolic PDE in  \eqref{IncompleteMarketProb_EarlyExercise_EqFormulation_Eq2}) by
\begin{itemize}
\item the penalty method devised in Sections \ref{Penalise_Discr_Obst_prob} and \ref{SectionModNewton_Isaacs} of this paper, and by
\item the method of policy iteration (see Appendix \ref{Appendix_PolicyIteration_MinMax} or \cite{Bokanowski_Howard_Algorithm}).
\end{itemize}
\end{itemize}
As already in Section \ref{SolvingDiscrSystems_ContControl}, we employ Remark \ref{Remark_StabilityInTheControl} and approximate $\mathbf{U}=[-1,0]$ by $\widetilde{\mathbf{U}}:=\{-1 + \frac{r}{101} : r = 0,1,\ldots 101 \}$. When solving the penalised equation \eqref{PenDiscreteProbDef_MinMax_Eq1} by Algorithm \ref{ModifiedNewtonAlg_Isaacs}, we use a test for accuracy of the kind \eqref{SolvingDiscrSystems_ContControl_Eq1}, setting $tol=$1e-08 as before. Similarly, we use a test for accuracy of the kind \eqref{SolvingDiscrSystems_ContControl_Eq0.9} with the same tolerance when solving an equation of the form \eqref{DiscreteProbDef_MinMax_Eq1} by policy iteration.

\begin{figure}[hb]
\centering
\includegraphics[width=11cm,height=8cm]{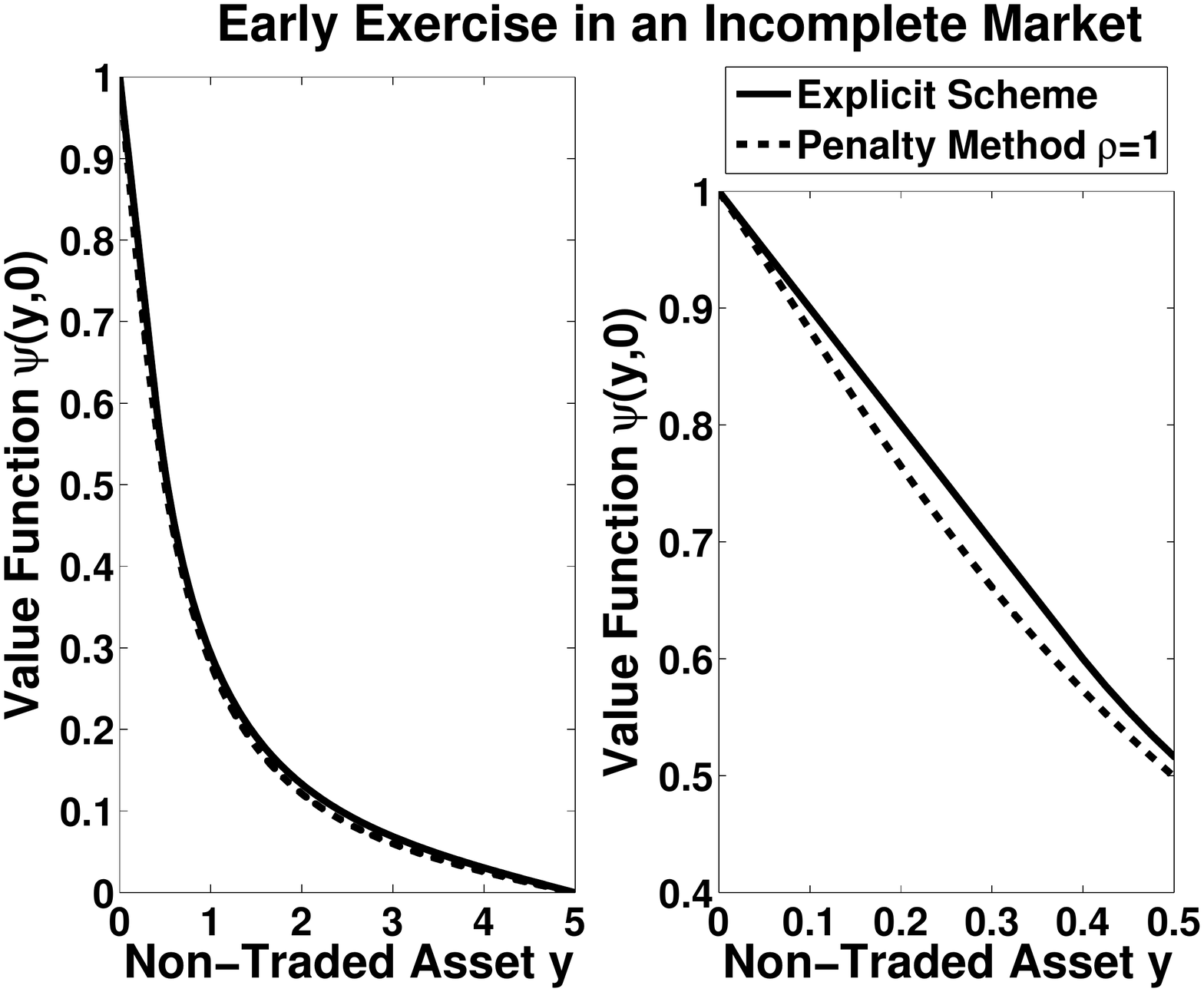}
\caption{The incomplete market early exercise pricing problem of Section \ref{Section_Examples_EarlyExercise}. We see the penalty approximation (for $M=N=200$ and $\rho=1$) and the solution
of the explicit scheme (for $M=4e04$ and $N=200$). Even though the penalty parameter is still
very small, the penalty solution already seems to be a reasonable approximation.} 
\label{fig:EarlyExercise_SolutionFunction}
\end{figure}

\begin{remark}\label{ConvergenceRemark_EarlyExerciseExample}
As already in Remark \ref{ConvergenceRemark_HJBExample}, convergence of the fully implicit discretisation of \eqref{IncompleteMarketProb_EarlyExercise_EqFormulation_Eq2} to the unique viscosity solution can be guaranteed by noting that we have stability, monotonicity, and consistency of the discretisation, and by using a strong comparison principle (cf.\,\cite{Barles_ErrorBoundsMonotoneApproxSchemes}). The fully explicit discretisation of \eqref{IncompleteMarketProb_EarlyExercise_EqFormulation_Eq1} converges similarly provided we have stability.
\end{remark}

In our numerical tests, we find the explicit scheme to require a relatively large number of time steps for stability, making it difficult to use. In Table \ref{tab:EarlyExercise_FirstNumbers}, fixing $\rho=$1e06, we see the difference between the explicit scheme and the penalty approximation for different grid sizes; for the explicit scheme, for a given space discretisation, we always choose the number of time steps such that the solution plot does not show any instabilities, whereas for the penalty scheme we take identical numbers of time and space steps. For $N=50$ and $N=200$, the explicit and the penalty scheme differ by $1.5$e-03 and $3.6$e-04 in the $max$-norm, respectively (cf.\,Table \ref{tab:EarlyExercise_FirstNumbers}). We point out that the explicit scheme runs substantially longer due to the high number of time steps required for stability; picking the time and space steps proportional to each other for the fully
implicit scheme is optimal experimentally because of the observed first
order convergence in both time and space.
In Figure \ref{fig:EarlyExercise_SolutionFunction}, we see the penalty approximation for $\rho=1$ and the explicit solution; even though the penalty parameter is very small, the two graphs are extremely close.
The difference in the $max$-norm between the policy iteration and penalty approximation solutions is $1.6165$e-05 and $2.6011$e-05 for grid sizes $M=N=50$ and $M=N=200$, respectively.
In Table \ref{tab:EarlyExercise_NewtonSteps}, for different grid sizes and penalty parameters, we see the maximum and the average number of iterations needed by Algorithms \ref{ModifiedNewtonAlg_Isaacs} (penalty approximation) and \ref{PolicyIteration_MinMax_Alg} (policy iteration) for solving the discrete systems at every time step, as well as the corresponding runtimes for the full schemes.
(In case of Algorithm \ref{PolicyIteration_MinMax_Alg}, we count one iteration whenever line \eqref{PolicyIteration_MinMax_Alg_Eq1} is executed.) Throughout, the average numbers of iterations are small, and both schemes run very fast, with the penalty scheme being faster by about a factor two; the effect appears to be due to the fact that -- whilst the average number of iterations is small -- policy iteration requires a large number of iterations in a few instances (as seen by the $Max\ Iterations$ in Table \ref{tab:EarlyExercise_NewtonSteps}); we will investigate this effect more closely below. Finally, in Figure \ref{fig:HJB_PenaltyConvergence}, we measure the  rate of convergence in $\rho$,
and confirm first order convergence as predicted by Theorem \ref{PenaltyConvergenceToTrueSol_ErrorEstimate_Isaacs}; the implementation is precisely as
in Section \ref{Section_Examples_Indifference}, except that we use a $\rho=1e08$ penalty approximation as reference solution; as before, the error is measured in the $max$-norm.

\begin{table}[t]
\begin{tabular}{|c|c||c|c||c|}
\hline
\textit{Explicit Scheme} & \textit{Time} & \textit{Penalty} $(\rho = 1e06)$ & \textit{Time} & \textit{Difference} \\
  \hline
$M=2500$, $N=50$ & 0.78s & $M=50$, $N=50$ & 0.23s & 1.3e-03\\
\hline
$M=4e04$, $N=200$ & 16.49s & $M=200$, $N=200$ & 3.71s & 3.5e-04\\
\hline
\end{tabular}
\caption{The incomplete market early exercise pricing problem of Section \ref{Section_Examples_EarlyExercise}. For different time and space grids, we see the difference between the explicit scheme and the penalty method, and the respective runtimes. The explicit scheme runs much longer due to the high number of time steps needed to guarantee stability for a given space discretisation. Furthermore,
since our space discretisation contains one-sided differences for reasons of monotonicity, the expected consistency order is $O(1/M)+O(1/N)$, which also makes the choice $M=N$ desirable.}
\label{tab:EarlyExercise_FirstNumbers}
\end{table}

\begin{table}[t]
\begin{tabular}{|l|c|c|c|c|c|c|}
\hline
\textit{Policy Iteration} & Max Iterations & $\varnothing$ Iterations & Runtime\\
  \hline
$M$, $N=50$ & 4 & 2.20 & 0.38s\\
  \hline
$M$, $N=200$ & 11& 2.17 & 6.42s\\
\hline
$M=200$, $N=50$ & 4 & 1.83 & 0.86s\\
\hline
$M=50$, $N=200$ & 18 & 2.88 & 2.12s\\
\hline
\hline
\textit{Penalty Method} $(\rho = 4e03)$ & Max Iterations & $\varnothing$ Iterations & Runtime\\
\hline
$M$, $N=50$ & 3 & 1.98 & 0.25s\\
  \hline
$M$, $N=200$ & 3 & 1.21 & 4.09s\\
\hline
$M=200$, $N=50$ & 3 & 1.15 & 0.66s\\
\hline
$M=50$, $N=200$ & 4 & 2.16 & 1.47s\\
\hline
\textit{Penalty Method} $(\rho = 1e06)$ & Max Iterations & $\varnothing$ Iterations & Runtime\\
\hline
\hline
$M$, $N=50$ & 2 & 1.10 & 0.17s\\
\hline
$M$, $N=200$ & 3 & 1.08 & 3.82s\\
\hline
$M=200$, $N=50$ & 2 & 1.02 & 0.61s\\
\hline
$M=50$, $N=200$ & 4 & 1.38 & 1.13s\\
\hline
\end{tabular}
\caption{The incomplete market early exercise pricing problem of Section \ref{Section_Examples_EarlyExercise}. For different time and space grids, we see the number of iterations needed
by Algorithm \ref{ModifiedNewtonAlg_Isaacs} when solving the penalised equation \eqref{PenDiscreteProbDef_MinMax}. Independently of the grid size, the absolute and the average number of required iterations is small.}
\label{tab:EarlyExercise_NewtonSteps}
\end{table}

\begin{figure}[ht]
\centering
\includegraphics[width=8cm,height=8cm]{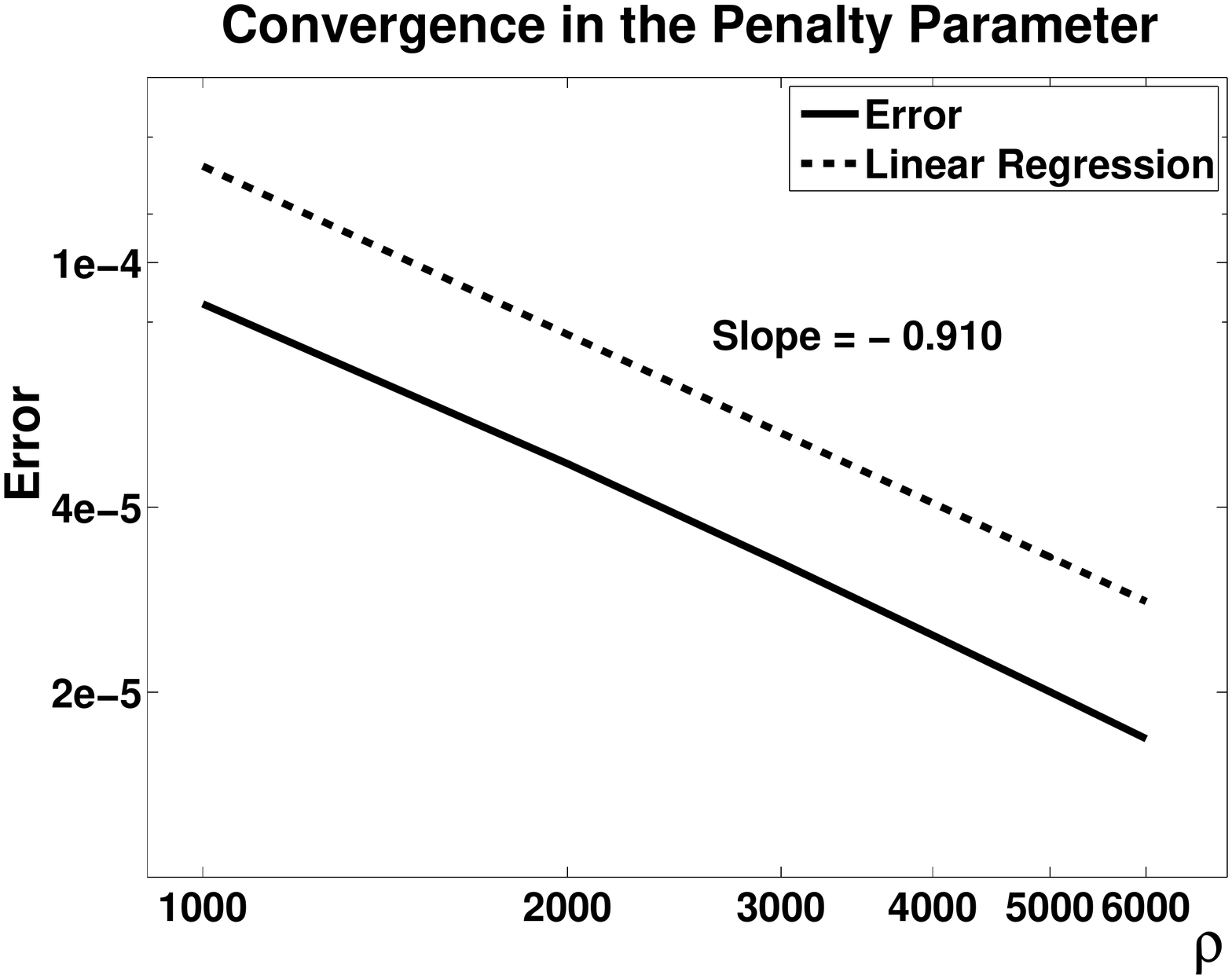}
\caption{The incomplete market early exercise pricing problem of Section \ref{Section_Examples_EarlyExercise}. For M=N=200, we measure the speed of convergence in the penalty parameter $\rho$. The error is measured in the $max$-norm. The plot is log-log, and we observe
a convergence rate of 0.910, close to one, confirming the results of Theorem \ref{PenaltyConvergenceToTrueSol_ErrorEstimate_Isaacs}.}
\label{fig:EarlyExercise_PenaltyConvergence}
\end{figure}

\subsubsection{Sensitivity with Respect to the Initial Guess}\label{InitialGuess_EarlyExercise}

We have seen above that, unlike penalty approximation, policy iteration requires many iterations in a few instances and that this effect appears to correlate with the grid size $N$ (cf.\,Table \ref{tab:EarlyExercise_NewtonSteps}). To investigate if the phenomenon relates
to the quality of the initial guess for the non-linear iterations, we set $M=1$ and consider different grids sizes $N$; the results can be seen in Figure \ref{fig:EarlyExercise_InitialGuessSensitivity}. (Since, in our implementations, we use the solution from the previous time step as initial guess for the next, setting $M=1$ can be interpreted as solving a single non-linear discrete system with a poor initial guess.) Clearly, the number of Newton iterations for the penalised system is almost unaffected by the increase of $N$, whereas the number of iterations of the policy iterations grows linearly in $N$. (In \cite{Bokanowski_Howard_Algorithm}, it has already been observed that even for a simple American option problem, the best obtainable bound on the number of iterations of policy iteration is linear, i.e. $O(N)$.)
\bigskip

Now, it is easy to see that equation \eqref{Appendix_PolicyIteration_MinMax_Eq1} is scalable, i.e., for any $\delta>0$, it can equivalently be rewritten as
\begin{equation*}
\min\Big\{ \max_{u\in\mathbf{U}}\{A_u\,z-b_u\},\delta(\tilde{A}z-\tilde{b})\Big\}=0,
\end{equation*}
and it has been pointed out in \cite{Forsyth_CombinedFixedPointIteration,Forsyth_RegimeSwitching,Forsyth_InexactArithmetic_DirectControl_PenaltyMethods} that a different choice of $\delta$ will generally lead to different policy iterations; more precisely, in our context, depending on the choice of $\delta$, we obtain
different adaptations of Algorithm \ref{PolicyIteration_MinMax_Alg}, all converging to the same solution. For theoretical considerations on the best choice of $\delta$, we refer to \cite{Forsyth_RegimeSwitching}; numerically, we find the following.
\begin{itemize}
\item Simply introducing a scaling factor $\delta$ does not yield an improvement.
\item Changing the initial guess from the payoff $P(\cdot)$ to $z^0\equiv 1$ does not yield an improvement.
\item Using a scaling factor $\delta=1$e06, combined with initial guess $z^0\equiv 1$, significantly reduces the number of iterations needed (cf.\,Figure \ref{fig:EarlyExercise_InitialGuessSensitivity}).
\end{itemize}
Further analysis shows that a large scaling factor $\delta$ yields an improvement whenever the initial guess is such that $\tilde{A}z^0-\tilde{b}>0$, which is necessary for the multiplication by the scaling factor to have an effect.\bigskip

In summary, we can conclude that the penalty approximation appears to have a generic advantage when dealing with poor starting values,
whereas -- to obtain equally good results by policy iteration -- prudent implementation is inevitable. The main reason for the different performance of policy iteration seems to be that it does $not$ show Newton-type behaviour, i.e. it does not converge to the
solution in steps with rapidly decreasing size; in particular, when using the payoff as initial guess, it shifts the solution upwards node by node until the free boundary is found, resulting in the linear dependence on $N$ observed in Figure \ref{fig:EarlyExercise_InitialGuessSensitivity}.

\begin{figure}[ht]
\centering
\includegraphics[width=8cm,height=8cm]{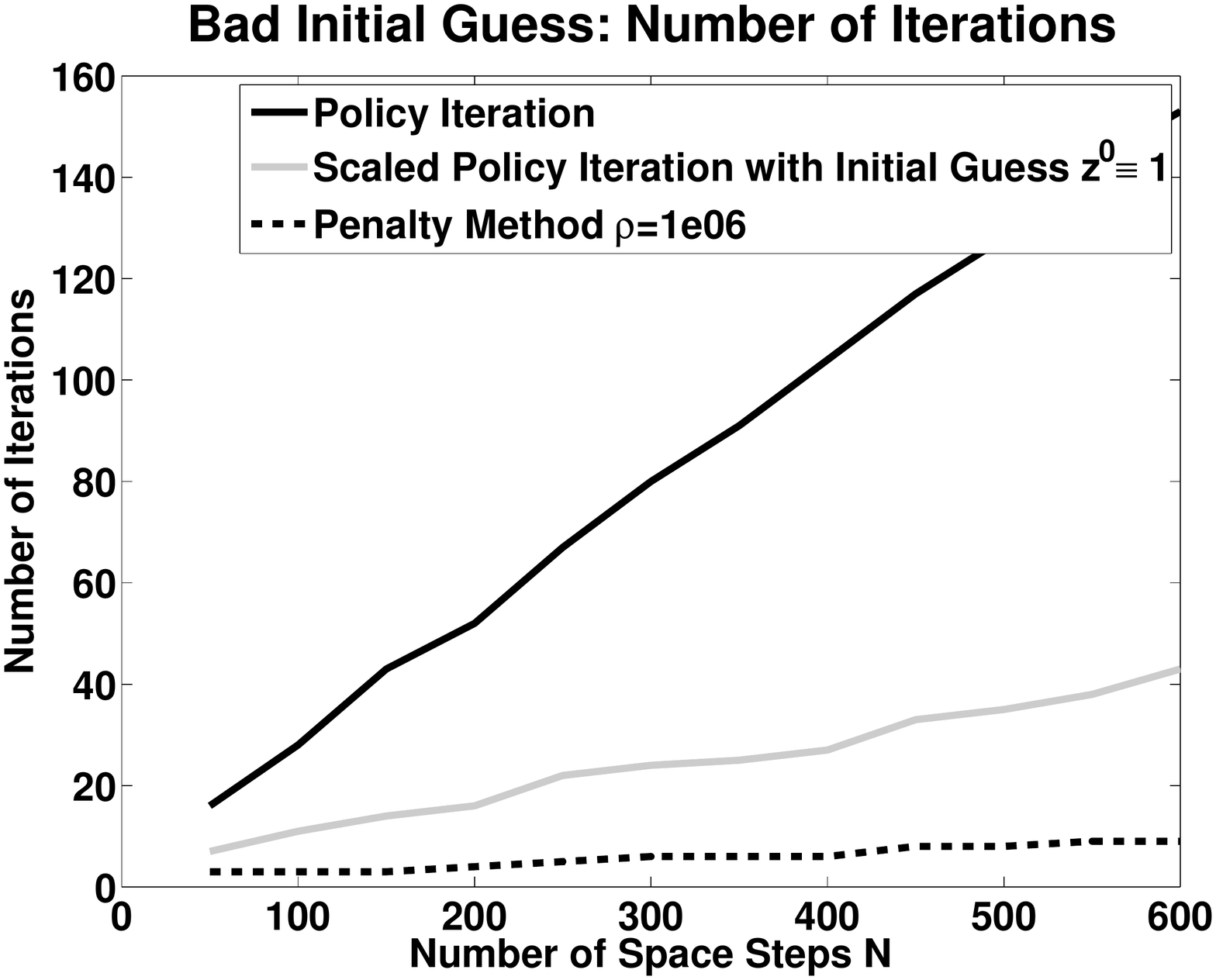}
\caption{The incomplete market early exercise pricing problem of Section \ref{Section_Examples_EarlyExercise}. For $M=1$ and varying $N$, we measure the number of non-linear iterations required by (scaled) policy iteration and penalty approximation, respectively.}
\label{fig:EarlyExercise_InitialGuessSensitivity}
\end{figure}

\section{Conclusion}
In this paper, we consider the numerical solution of continuously controlled HJB equations and HJB obstacle problems -- which trivially includes finitely controlled equations -- and we show that penalisation is a powerful means
for solving the non-linear discrete problems resulting from implicit finite difference discretisations.
Generally, this can be done by policy iteration or -- as we show -- by penalisation combined with a Newton-type iteration. For both penalty approaches, we show that
the achieved accuracy is $O(1/\rho)$, where $\rho$ is the penalty parameter.
We include numerical examples
from (early exercise) incomplete market pricing, demonstrating the competitiveness of our algorithms as fast and easy-to-use numerical schemes.
An interesting open problem is the extension of our approach to more general Isaacs equations.

\appendix

\section{Policy Iteration for the HJB Equation}\label{Appendix_PolicyIteration_ContinuouslyControlled}

We briefly recap the policy iteration algorithm for HJB equations as introduced in \cite{Forsyth_Controlled_HJB_PDEs_Finance} and \cite{Bokanowski_Howard_Algorithm}. Recall that we want to solve Problem \ref{DiscreteProbDef}, i.e. we are trying to find $x\in\mathbb{R}^N$ such that
\begin{equation}
\min\{A_u\,x-b_u : u\in\mathbf{U}\}=0.\label{PolicyIteration_Alg_Eq0.9}
\end{equation}

\begin{alg}\label{PolicyIteration_Alg}(Policy Iteration for HJB Eq.)
For $i\in\mathcal{N}$ and $y\in\mathbb{R}^N$, define
$u^{min}_i(y)\in\mathbf{U}$ to be such that
\begin{equation*}
u^{min}_i(y) = \argmin\{(A_v\,y-b_v)_i : v\in\mathbf{U}\},
\end{equation*}
and set $A^{min}(y)\in\mathbb{R}^{N\times N}$ and $b^{min}(y)\in\mathbb{R}^N$ to be matrix and vector consisting of
rows $(A_{u^{min}_i(y)})_i$ and $(b_{u^{min}_i(y)})_i$\,, $i\in\mathcal{N}$, respectively.
Let $x^0\in\mathbb{R}^N$ be some starting value. Then, for known $x^n$, $n\geq 0$, find $x^{n+1}$ such that
\begin{equation}
A^{min}(x^n)\,x^{n+1} = b^{min}(x^n).\label{PolicyIteration_Alg_Eq1}
\end{equation}
\end{alg}

\begin{theorem}\label{PolicyIteration_ConvergenceTheorem}
Let $(x^n)_{n\geq 0}$ be the sequence generated by Algorithm \ref{PolicyIteration_Alg}.
We have $x^{n+1}\geq z^n$ for $n\geq 1$. As $n\to \infty$, $x^n$ converges to a limit $z^*$ which solves \eqref{PolicyIteration_Alg_Eq0.9}.
\end{theorem}
\begin{proof}
See \cite{Forsyth_Controlled_HJB_PDEs_Finance} or \cite{Bokanowski_Howard_Algorithm}.
\end{proof}

\section{Policy Iteration for the HJB Obstacle Problem}\label{Appendix_PolicyIteration_MinMax}

We briefly recap the method of policy iteration algorithm for HJB obstacle problems that was introduced in \cite{Bokanowski_Howard_Algorithm}. Recall that we want to solve Problem \ref{DiscreteProbDef_MinMax}, i.e. we are trying to find $z\in\mathbb{R}^N$ such that
\begin{equation}
\min\Big\{ \max_{u\in\mathbf{U}}\{A_u\,z-b_u\},\,\tilde{A}z-\tilde{b}\Big\}=0.\label{Appendix_PolicyIteration_MinMax_Eq1}
\end{equation}

Now, for $u\in\mathbf{U}$, we define $A_{u,0}:=A_u$ and $A_{u,1}:=\tilde{A}$, and we define $b_{u,0}$ and $b_{u,1}$ correspondingly.
Using these new definitions, \eqref{Appendix_PolicyIteration_MinMax_Eq1} is equivalent to
\begin{equation*}
\min_{v\in\{0,1\}}\Big\{ \max_{u\in\mathbf{U}}\{A_{u,v}\,z-b_{u,v}\}\Big\}=0
\end{equation*}

\begin{alg}\label{PolicyIteration_MinMax_Alg}(Policy Iteration for HJB Obstacle Prob.)
For $i\in\mathcal{N}$ and $z\in\mathbb{R}^N$, define $v_i^{min}(z)\in\{0,1\}$ such that
\begin{equation*}
\max_{u\in\mathbf{U}}(A_{u,v_i^{min}(z)}\,z-b_{u,v_i^{min}(z)})_i=\min_{v\in\{0,1\}}\Big\{\max_{u\in\mathbf{U}}(A_{u,v}\,z-b_{u,v})_i\Big\}.
\end{equation*}
Let $z^0\in\mathbb{R}^N$ be some starting value. Then, for known $z^n$, $n\geq 0$, find $z^{n+1}$ such that
\begin{equation}
\max_{u\in\mathbf{U}}(A_{u,v_i^{min}(z^n)}\,z^{n+1}-b_{u,v_i^{min}(z)})_i=0,\quad i\in\mathcal{N}.\label{PolicyIteration_MinMax_Alg_Eq1}
\end{equation}
\end{alg}

\begin{theorem}\label{PolicyIteration_MinMax_ConvergenceTheorem}
Let $(z^n)_{n\geq 0}$ be the sequence generated by Algorithm \ref{PolicyIteration_MinMax_Alg}.
We have $z^{n+1}\geq z^n$ for $n\geq 1$. As $n\to \infty$, $z^n$ converges to a limit $z^*$ which solves \eqref{Appendix_PolicyIteration_MinMax_Eq1}.
\end{theorem}
\begin{proof}
See \cite{Bokanowski_Howard_Algorithm}.  
\end{proof}

\renewcommand{\bibname}{References}

\bibliography{Paper_PenSchemeDiscrHJB_InfinitelyManyControls_References}
\bibliographystyle{plain}

\end{document}